\documentclass[12pt]{article}%
\usepackage{amssymb}
\usepackage{amsfonts}
\usepackage{amsmath}
\usepackage{graphicx}%
\providecommand{\U}[1]{\protect\rule{.1in}{.1in}}
\newtheorem{theorem}{Theorem}

\newtheorem{definition}[theorem]{Definition}

\newtheorem{proposition}[theorem]{Proposition}
\newtheorem{remark}[theorem]{Remark}

\newenvironment{proof}[1][Proof]{\noindent\textbf{#1.} }{\ \rule{0.5em}{0.5em}}

\begin{document}

\title{\textbf{A Robust Generalization of the Rao Test}}
\author{Ayanendranath Basu$^{1}$, Abhik Ghosh$^{1}$, Nirian Martin$^{2}$%
\thanks{Corresponding author; Email: nirian@estad.ucm.es} and Leandro
Pardo$^{2}\medskip$\\$^{1}${\small Indian Statistical Institute, Kolkata, India }\\$^{2}${\small Complutense University, Madrid, Spain }}
\date{}
\maketitle

\begin{abstract}
This paper presents new families of Rao-type test statistics based on the
minimum density power divergence estimators which provide robust
generalizations\ for testing simple and composite null hypotheses. The
asymptotic null distributions of the proposed tests are obtained and their
robustness properties are also theoretically studied. Numerical illustrations are provided to substantiate the theory developed. On the whole, the proposed tests are seen
to be excellent alternatives to the classical Rao test.

\end{abstract}




\section{Introduction}

The systematic use of hypothesis testing, which originated with the
publication of Pearson's (1900) goodness-of-fit test paper, was theoretically
formalized by Neyman and Pearson (1928, 1933). In the first paper they
introduced the likelihood ratio (LR) test and in the second paper they
established a general principle for constructing optimal tests. Wilks (1938)
obtained the asymptotic distribution of the LR test. Later Wald (1943)
introduced a test procedure which is commonly known as the Wald test. Rao
(1948) introduced the Rao test (or the score test) as an alternative to the LR
and the Wald tests. Aitchison and Silvey (1958) and Silvey (1959) gave an
interpretation of the Rao test in terms of the Lagrange multiplier test.

The tests described in the previous paragraph are all based on the maximum
likelihood estimator (MLE). The non-robust nature of the procedures based on
the MLE has motivated several researchers to look for robust alternatives to
the LR, Wald and Rao tests. We are interested, in particular, in the robust,
minimum divergence branch of this approach, where the non-robust MLE is
replaced by a suitable robust minimum divergence estimator. In this line of
research we encounter, among other approaches, the Wald-type test statistics
based on minimum density power divergence estimators (MDPDEs) considered, for
instance, in Basu et al.~(2016) and Ghosh et al.~(2016).

This paper aims at providing a robust generalization of the Rao test. The
score test is a popular tool in statistics, and there is a significant body of
research that deals with the derivation and implementation of many different
score test statistics, the reinterpretation of other approaches to testing as
variants of the score method, and the properties of score tests in
non-standard situations. There are several books and survey articles available
on this topic, e.g. Bera and Ullah (1991), Breusch and Pagan (1980), Engle
(1984), Godfrey (1988), Godfrey and Tremayne (1988), Kramer and Sonnberger
(1986), Maddala (1995) and White (1984). A nice review of the literature
concerning the Rao test is available in Rao (2005).

As indicated, most of the papers and books published in relation to the Rao
test are based on MLEs. It is true that in case of the Rao test for the simple
null hypothesis no parameter estimation is necessary, but the Rao test in that
case is based on the likelihood score function associated with the MLE. Our
purpose in this paper is to present Rao-type test statistics for testing
simple and composite null hypothesis based on the MDPDE and the corresponding
estimating functions. This estimator was proposed by Basu et al.~(1998), and
appears to have significantly affected future research; it exhibits excellent
behavior in terms of combining high model efficiency with strong outlier
robustness. In this paper we will demonstrate that similar nice properties
carry over in the domain of parametric hypothesis testing.

The rest of the paper is organized as follows. In Section 2 we present the
MDPDE as well as the restricted minimum density power divergence estimator
(RMDPDE). The Rao-type tests for the simple null hypothesis are introduced in
Section 3; its asymptotic properties under the null hypothesis as well as
under contiguous alternative hypotheses are also studied in this section. The
case of the composite null hypotheses are presented in Section 4. Some
illustrative examples are presented in Section 6, while the results of an
extensive simulation study are described in Section 7.


\section{The density power divergence}


The density power divergence family (Basu et al.,~1998) represents a rich
class of density based divergences including the Kullback-Leibler divergence.
Consider the class of distributions having densities with respect to a given
dominating measure. Let $\mathcal{G}$ represent the class of corresponding
densities. Given two densities $g,f\in\mathcal{G}$, the density power
divergence between them is defined, as a function of a nonnegative tuning
parameter $\beta$, through the relation
\begin{equation}
d_{\beta}(g,f)=\left\{
\begin{array}
[c]{ll}%
\int_{-\infty}^{+\infty}\left\{  f^{\beta+1}(x)-\tfrac{\beta+1}{\beta}%
f^{\beta}(x)g(x)+\frac{1}{\beta}g^{\beta+1}(x)\right\}  dx, & \text{for}%
\mathrm{~}\beta>0,\\[2ex]%
\int_{-\infty}^{+\infty}g(x)\log\left(  \frac{g(x)}{f(x)}\right)  dx, &
\text{for}\mathrm{~}\beta=0.
\end{array}
\right.  \label{EQ:definition_DPD}%
\end{equation}
We are assuming a univariate random variable associated either with $f$ or
$g$, but it can be generalized to multivariate case without any loss of
generality. The case corresponding to $\beta=0$ may be derived from the
general case by taking the continuous limit as $\beta$ tends to $0$. The
quantities defined in Equation (\ref{EQ:definition_DPD}) are genuine
divergences in the sense\ that $d_{\beta}(g,f)\geq0$ for all $g,f\in
\mathcal{G}$ and all $\beta\geq0$, and $d_{\beta}(g,f)$ is equal to zero if
and only if the densities $g$ and $f$ are identically equal.

We consider the parametric model of densities $\{f_{\boldsymbol{\theta}%
}:\boldsymbol{\theta}\in\Theta\subset{\mathbb{R}}^{p}\}$; we are interested in
the estimation of the parameter $\boldsymbol{\theta}$. Let $G$ represent the
distribution function corresponding to the density function $g$. The MDPD
functional $\boldsymbol{T}_{\beta}(G)$ at $G$ is defined by the requirement
$d_{\beta}(g,f_{\boldsymbol{T}_{\beta}(G)})=\min_{\boldsymbol{\theta}\in
\Theta}d_{\beta}(g,f_{\boldsymbol{\theta}})$. Clearly, the term $\int%
_{-\infty}^{+\infty}g^{\beta+1}(x)dx$ has no role in the minimization of
$d_{\beta}(g,f_{\boldsymbol{\theta}})$ over $\boldsymbol{\theta}\in\Theta$.
Therefore the objective function to be minimized in the computation of the
MDPD functional $\boldsymbol{T}_{\beta}(G)$ simplifies to
\[
\int_{-\infty}^{+\infty}\left\{  f_{\boldsymbol{\theta}}^{\beta+1}%
(x)-\tfrac{\beta+1}{\beta}f_{\boldsymbol{\theta}}^{\beta}(x)g(x)\right\}
dx=\int_{-\infty}^{+\infty}f_{\boldsymbol{\theta}}^{\beta+1}(x)dx-\tfrac
{\beta+1}{\beta}\int_{-\infty}^{+\infty}f_{\boldsymbol{\theta}}^{\beta
}(x)dG(x).
\]
Further, given a random sample $X_{1},\ldots,X_{n}$ from the distribution $G$
we can approximate the second term in the objective function by replacing $G$
with its empirical estimate $G_{n}$. For a given tuning parameter $\beta$,
therefore, the MDPDE $\widehat{\boldsymbol{\theta}}_{\beta}$ of
$\boldsymbol{\theta}$ can be obtained by minimizing
\begin{equation}
\int_{-\infty}^{+\infty}f_{\boldsymbol{\theta}}^{\beta+1}(x)dx-\tfrac{\beta
+1}{\beta}\int_{-\infty}^{+\infty}f_{\boldsymbol{\theta}}^{\beta}%
(x)dG_{n}(x)=\int_{-\infty}^{+\infty}f_{\boldsymbol{\theta}}^{\beta
+1}(x)dx-\tfrac{\beta+1}{\beta}\frac{1}{n}\sum_{i=1}^{n}f_{\boldsymbol{\theta
}}^{\beta}(X_{i}) \label{EQ:Vtheta}%
\end{equation}
over $\boldsymbol{\theta}\in\Theta$. The remarkable observation here is that
the minimization of the above expression over $\boldsymbol{\theta}$ does not
require the use of a non-parametric density estimate.


Under differentiability of the model the minimization of the objective
function in Equation (\ref{EQ:Vtheta}) leads to an estimating equation of the
form
\begin{equation}
\frac{1}{n}\sum_{i=1}^{n}\boldsymbol{s}_{\boldsymbol{\theta}}(X_{i}%
)f_{\boldsymbol{\theta}}^{\beta}(X_{i})-\int_{-\infty}^{+\infty}%
\boldsymbol{s}_{\boldsymbol{\theta}}(x)f_{\boldsymbol{\theta}}^{\beta
+1}(x)dx=\boldsymbol{0}_{p}, \label{EQ:general_case}%
\end{equation}
where $\boldsymbol{s}_{\boldsymbol{\theta}}(x)=\frac{\partial}{\partial
\boldsymbol{\theta}}\log f_{\boldsymbol{\theta}}(x)$ is the likelihood score
function of the model. It may also be noted that Equation
(\ref{EQ:general_case}) is an unbiased estimating equation under the model.
Since the above estimating equation weights the score $\boldsymbol{s}%
_{\boldsymbol{\theta}}(X_{i})$ with the power of the density
$f_{\boldsymbol{\theta}}^{\beta}(X_{i})$, it is obvious that score functions
for the observations that are discrepant with respect to the model will be
downweighted by the density component.

The functional $\boldsymbol{T}_{\beta}(G)$ is Fisher consistent; it takes the
value $\boldsymbol{\theta}_{0}$ when the true density $g=f_{\boldsymbol{\theta
}_{0}}$ is in the model. When it is not, $\boldsymbol{\theta}_{\beta}%
^{g}=\boldsymbol{T}_{\beta}(G)$ represents the best fitting parameter.
Suppressing the $\beta$ subscript in the notation for $\boldsymbol{\theta
}_{\beta}^{g}$, $f_{\boldsymbol{\theta}^{g}}$ is seen to represent the model
element closest to the density $g$ in the density power divergence sense.

Let $g$ be the true data generating density function and ${\boldsymbol{\theta
}}^{g}=\boldsymbol{T}_{\beta}(G)$ be the best fitting parameter. To set up the
notation we define the quantities
\begin{align}
\boldsymbol{J}_{\beta}(\boldsymbol{\theta})  &  =\int_{-\infty}^{+\infty
}\boldsymbol{s}_{\boldsymbol{\theta}}(x)\boldsymbol{s}_{{\boldsymbol{\theta}}%
}^{T}(x)f_{{\boldsymbol{\theta}}}^{\beta+1}(x)dx\nonumber\\
&  +\int_{-\infty}^{+\infty}\{\mathbb{I}_{\boldsymbol{\theta}}(x)-\beta
\boldsymbol{s}_{\boldsymbol{\theta}}(x)\boldsymbol{s}_{\boldsymbol{\theta}%
}^{T}(x)\}\{g(x)-f_{\boldsymbol{\theta}}(x)\}f_{{\boldsymbol{\theta}}}^{\beta
}(x)dx,\label{a2_new}\\
\boldsymbol{K}_{\beta}(\boldsymbol{\theta})  &  =\int_{-\infty}^{+\infty
}\boldsymbol{s}_{{\boldsymbol{\theta}}}(x)\boldsymbol{s}_{{\boldsymbol{\theta
}}}^{T}(x)f_{{\boldsymbol{\theta}}}^{2\beta}(x)g(x)dx-\boldsymbol{\xi}_{\beta
}({{\boldsymbol{\theta}}})\boldsymbol{\xi}_{\beta}^{T}({{\boldsymbol{\theta}}%
}), \label{a3_new}%
\end{align}
where $\boldsymbol{\xi}_{\beta}({{\boldsymbol{\theta}}})=\int_{-\infty
}^{+\infty}\boldsymbol{s}_{\boldsymbol{\theta}}(x)f_{\boldsymbol{\theta}%
}^{\beta}(x)g(x)dx$, and $\mathbb{I}_{\boldsymbol{\theta}}(x)=-\frac{\partial
}{\partial\boldsymbol{\theta}}\boldsymbol{s}_{\boldsymbol{\theta}}^{T}(x)$.

The following Basu et al. regularity conditions\textbf{ }(Basu et al., 2011)
form the basis of our subsequent developments.

\begin{itemize}
\item[(D1)] The model densities $f_{\boldsymbol{\theta}}$ and the true density
$g$ all have common support $\mathcal{X}$, which is independent of
$\boldsymbol{\theta}\in\Theta\subset%
\mathbb{R}
^{p}$.

\item[(D2)] There is an open $\Omega\subset\Theta$, containing the best
fitting parameter $\boldsymbol{\theta}^{g}$ under the true density $g$ such
that $f_{\boldsymbol{\theta}}(x)$ is thrice continuously differentiable with
respect to $\boldsymbol{\theta}\in\Omega$, for almost all $x\in\mathcal{X}$.

\item[(D3)] The integrals $\int_{-\infty}^{+\infty}f_{\boldsymbol{\theta}%
}^{\beta+1}(x)dx$ and $\int_{-\infty}^{+\infty}f_{\boldsymbol{\theta}}^{\beta
}(x)g(x)dx$ can be differentiated thrice in $\boldsymbol{\theta}$, and the
derivatives can be interchanged with the integration.

\item[(D4)] The matrix $\boldsymbol{J}_{\beta}(\boldsymbol{\theta})$, defined
in (\ref{a2_new}), is positive definite.

\item[(D5)] For each index $j,k,l=1,\ldots,p$, of\ single components in
$\boldsymbol{\theta}$, there exists\ a function $\zeta_{jkl,\beta}(x)$, having
finite expectation under $g$, which satisfies
\[
\left\vert \frac{\partial^{3}}{\partial\theta_{j}\partial\theta_{k}%
\partial\theta_{l}}\int_{-\infty}^{+\infty}\left[  f_{\boldsymbol{\theta}%
}^{1+\beta}-(1+1/\beta)f_{\boldsymbol{\theta}}^{\beta}(x)\right]
dx\right\vert \leq\zeta_{jkl,\beta}(x)\text{~for~all~}\boldsymbol{\theta}%
\in\Omega.
\]

\end{itemize}

\begin{theorem}
\label{Theorem1}(Basu et al., 2011) We assume that the Basu et al.~conditions
are true. Then,

\begin{enumerate}
\item[a)] the estimating equation (\ref{EQ:general_case}) for the MDPDE has a
consistent sequence of roots $\widehat{\boldsymbol{\theta}}_{\beta
}=\widehat{\boldsymbol{\theta}}_{n,\beta}$, i.e. $\widehat{\boldsymbol{\theta
}}_{n,\beta}\underset{n\rightarrow\infty}{\overset{P}{\longrightarrow}%
}\boldsymbol{\theta}_{0}$.

\item[b)] $n^{1/2}(\hat{\boldsymbol{\theta}}_{\beta}-{\boldsymbol{\theta}}%
^{g})$ has an asymptotic multivariate normal distribution with vector mean
zero and covariance matrix $\boldsymbol{J}_{\beta}^{-1}({{\boldsymbol{\theta}%
}}^{g})\boldsymbol{K}_{\beta}({{\boldsymbol{\theta}}}^{g})\boldsymbol{J}%
_{\beta}^{-1}({{\boldsymbol{\theta}}}^{g})$, where $\boldsymbol{J}_{\beta
}({{\boldsymbol{\theta}}}^{g})$ and $\boldsymbol{K}_{\beta}%
({{\boldsymbol{\theta}}}^{g})$ are as in (\ref{a2_new}) and (\ref{a3_new}).
\end{enumerate}
\end{theorem}

When the true distribution $G$ belongs to the model so that
$G=F_{\boldsymbol{\theta}}$ for some $\boldsymbol{\theta}\in\Theta$, the
formula for $\boldsymbol{J}$, $\boldsymbol{K}$ and $\boldsymbol{\xi}$ simplify
to
\begin{align}
\boldsymbol{J}_{\beta}(\boldsymbol{\theta})  &  =\int_{-\infty}^{+\infty
}\boldsymbol{s}_{\boldsymbol{\theta}}(x)\boldsymbol{s}_{\boldsymbol{\theta}%
}^{T}(x)f_{\boldsymbol{\theta}}^{\beta+1}(x)dx,\label{model_variance}\\
\boldsymbol{K}_{\beta}(\boldsymbol{\theta})  &  =\int_{-\infty}^{+\infty
}\boldsymbol{s}_{\boldsymbol{\theta}}(x)\boldsymbol{s}_{\boldsymbol{\theta}%
}^{T}(x)f_{\boldsymbol{\theta}}^{2\beta+1}(x)dx-\boldsymbol{\xi}_{\beta
}({{\boldsymbol{\theta}}})\boldsymbol{\xi}_{\beta}^{T}({{\boldsymbol{\theta}}%
}),\label{m_v}\\
\boldsymbol{\xi}_{\beta}(\boldsymbol{\theta})  &  =\int_{-\infty}^{+\infty
}\boldsymbol{s}_{\boldsymbol{\theta}}(x)f_{\boldsymbol{\theta}}^{\beta
+1}(x)dx. \label{xi_function}%
\end{align}

Consider a general composite null hypothesis of interest which restricts the
parameter space to a proper subset $\Theta_{0}$ of $\Theta$, i.e.%
\begin{equation}
H_{0}:\boldsymbol{\theta}\in\Theta_{0}~\text{against}~H_{1}:\boldsymbol{\theta
}\notin\Theta_{0}. \label{1}%
\end{equation}
In many practical hypothesis testing problems, the restricted parameter space
$\Theta_{0}$ is defined by a set of $r$ restrictions of the form
\begin{equation}
\boldsymbol{m}(\boldsymbol{\theta)=0}_{r} \label{0}%
\end{equation}
on $\Theta$, where $\boldsymbol{m}:{\mathbb{R}}^{p}\longrightarrow{\mathbb{R}%
}^{r}$ is a vector-valued function such that the $p\times r$ matrix
\begin{equation}
\boldsymbol{M}\left(  \boldsymbol{\theta}\right)  =\frac{\partial
\boldsymbol{m}^{T}(\boldsymbol{\theta)}}{\partial\boldsymbol{\theta}}
\label{0.0}%
\end{equation}
exists, $r\leq q$ and is continuous in $\boldsymbol{\theta}$ and
$\boldsymbol{M}\left(  \boldsymbol{\theta}\right)  $ is of full rank
($\mathrm{rank}\left(  \boldsymbol{M}\left(  \boldsymbol{\theta}\right)
\right)  =r)$. The superscript $T$ in the above represents the transpose of
the matrix. The RMDPD functional $\boldsymbol{T}_{\beta}^{0}(G)$ at $G$, on
the other hand, is the value in the parameter space which satisfies
\[
d_{\beta}(g,f_{\boldsymbol{T}_{\beta}^{0}(G)})=\min_{{{\boldsymbol{\theta}}%
}\in\Theta_{0}}d_{\beta}(g,f_{{\boldsymbol{\theta}}}),
\]
given such a minimizer exists. When a random sample $X_{1},\ldots,X_{n}$ is
available from the distribution $G$, the RMDPDE of $\boldsymbol{\theta}$
minimizes the objective function in (\ref{EQ:Vtheta}) subject to
$\boldsymbol{m}(\boldsymbol{\theta})=\boldsymbol{0}_{r}$. Under this set the
next theorem presents the asymptotic distribution of the RMDPDE
$\boldsymbol{\tilde{\theta}}_{\beta}$ of $\boldsymbol{\theta}$.

\begin{theorem}
(Basu et al., 2014) The RMDPDE $\boldsymbol{\tilde{\theta}}_{\beta}$ of
$\boldsymbol{\theta}$ obtained under the null hypothesis through the
constraints $\boldsymbol{m}(\boldsymbol{\theta})=\boldsymbol{0}_{r}$, has the
asymptotic distribution
\[
n^{1/2}(\boldsymbol{\tilde{\theta}}_{\beta}-\boldsymbol{\theta}_{0}%
)\underset{n\rightarrow\infty}{\overset{\mathcal{L}}{\longrightarrow}%
}\mathcal{N}(\boldsymbol{0}_{p},\boldsymbol{\Sigma}_{\beta}(\boldsymbol{\theta
}_{0})\boldsymbol{)}%
\]
where
\[
\boldsymbol{\Sigma}_{\beta}(\boldsymbol{\theta}_{0})=\boldsymbol{P}_{\beta
}(\boldsymbol{\theta}_{0})\boldsymbol{K}_{\beta}(\boldsymbol{\theta}%
_{0})\boldsymbol{P}_{\beta}(\boldsymbol{\theta}_{0}),
\]
\begin{equation}
\boldsymbol{P}_{\beta}(\boldsymbol{\theta}_{0})=\boldsymbol{J}_{\beta}%
^{-1}(\boldsymbol{\theta}_{0})-\boldsymbol{Q}_{\beta}(\boldsymbol{\theta}%
_{0})\boldsymbol{M}^{T}(\boldsymbol{\theta}_{0})\boldsymbol{J}_{\beta}%
^{-1}(\boldsymbol{\theta}_{0}), \label{P}%
\end{equation}
$\boldsymbol{M}(\boldsymbol{\theta}_{0})$, $\boldsymbol{J}(\boldsymbol{\theta
}_{0})$\ are as in (\ref{0.0}) and (\ref{a2_new}) respectively, and evaluated
at $\boldsymbol{\theta}=\boldsymbol{\theta}_{0}$, and
\begin{equation}
\boldsymbol{Q}_{\beta}(\boldsymbol{\theta}_{0})=\boldsymbol{J}_{\beta}%
^{-1}(\boldsymbol{\theta}_{0})\boldsymbol{M}(\boldsymbol{\theta}_{0})\left[
\boldsymbol{M}^{T}(\boldsymbol{\theta}_{0})\boldsymbol{J}_{\beta}%
^{-1}(\boldsymbol{\theta}_{0})\boldsymbol{M}(\boldsymbol{\theta}_{0})\right]
^{-1}. \label{Q}%
\end{equation}

\end{theorem}


\section{Rao-type statistics for testing simple null hypothesis}


The MDPDE can be obtained solving the system of estimating equations
\begin{equation}
\boldsymbol{U}_{\beta,n}\left(  \boldsymbol{\theta}\right)  =\boldsymbol{0}%
_{p}, \label{EQ:score_equation}%
\end{equation}
where
\begin{align}
\boldsymbol{U}_{\beta,n}\left(  \boldsymbol{\theta}\right)   &  =\frac{1}%
{n}\sum_{i=1}^{n}\boldsymbol{u}_{\beta}\left(  X_{i},\boldsymbol{\theta
}\right)  ,\label{2.1}\\
\boldsymbol{u}_{\beta}\left(  x,\boldsymbol{\theta}\right)   &
=\boldsymbol{s}_{\boldsymbol{\theta}}(x)f_{\boldsymbol{\theta}}^{\beta
}(x)-\int_{-\infty}^{+\infty}\boldsymbol{s}_{\boldsymbol{\theta}%
}(y)f_{\boldsymbol{\theta}}^{\beta+1}(y)dy \label{EQ:score_function}%
\end{align}
is the $\beta$-score statistic.
Notice that denoting $\boldsymbol{u}_{\beta}\left(  x,\boldsymbol{\theta
}\right)  =(u_{1,\beta}\left(  x,\boldsymbol{\theta}\right)  ,...,u_{p,\beta
}\left(  x,\boldsymbol{\theta}\right)  )^{T}$, its correspondent components
are
\[
\boldsymbol{U}_{\beta,n}\left(  \boldsymbol{\theta}\right)  =\left(  \frac
{1}{n}\sum_{i=1}^{n}u_{1,\beta}\left(  X_{i},\boldsymbol{\theta}\right)
,...,\frac{1}{n}\sum_{i=1}^{n}u_{p,\beta}\left(  X_{i},\boldsymbol{\theta
}\right)  \right)  ^{T}.
\]
Simple calculations show that
\[
E\left[  \boldsymbol{U}_{\beta,n}\left(  \boldsymbol{\theta}\right)  \right]
=\boldsymbol{0}_{p},~Cov\left[  \boldsymbol{U}_{\beta,n}\left(
\boldsymbol{\theta}\right)  \right]  =\frac{1}{n}\boldsymbol{K}_{\beta}\left(
\boldsymbol{\theta}\right)  ,
\]
and we have, by the Central Limit Theorem (CLT)
\begin{equation}
\sqrt{n}\boldsymbol{U}_{\beta,n}\left(  \boldsymbol{\theta}\right)
\underset{n\rightarrow\infty}{\overset{\mathcal{L}}{\longrightarrow}%
}\mathcal{N}\left(  \boldsymbol{0},\boldsymbol{K}_{\beta}\left(
\boldsymbol{\theta}\right)  \right)  , \label{2.0}%
\end{equation}
where $\boldsymbol{K}_{\beta}\left(  \boldsymbol{\theta}\right)  $ is as in
Equation (\ref{m_v}).
In this setting is introduced the Rao-type test statistics of order $\beta
$\ for testing the simple null hypothesis
\begin{equation}
H_{0}:\boldsymbol{\theta=\theta}_{0}\text{~against~}H_{1}:\boldsymbol{\theta
}\neq\boldsymbol{\theta}_{0}. \label{2.2}%
\end{equation}

\begin{definition}
The Rao-type test statistic of order $\beta$ for testing the null hypothesis
in (\ref{2.2}) is given by
\begin{equation}
R_{\beta,n}\left(  \boldsymbol{\theta}_{0}\right)  =n\boldsymbol{U}_{\beta
,n}^{T}\left(  \boldsymbol{\theta}_{0}\right)  \boldsymbol{K}_{\beta}%
^{-1}\left(  \boldsymbol{\theta}_{0}\right)  \boldsymbol{U}_{\beta,n}\left(
\boldsymbol{\theta}_{0}\right)  \label{2.3}%
\end{equation}
where $\boldsymbol{U}_{\beta,n}\left(  \boldsymbol{\theta}_{0}\right)  $ is as
defined in (\ref{2.1}).
\end{definition}

Given Equations (\ref{2.0}) and (\ref{2.1}), the asymptotic distribution of
the Rao-type test statistics in Equation (\ref{2.3}) can be easily derived,
which is stated in the following theorem.

\begin{theorem}
The asymptotic distribution of the Rao-type test statistics $R_{\beta
,n}\left(  \boldsymbol{\theta}_{0}\right)  $ given in (\ref{2.3}), is
chi-square with $p$ degrees of freedom under the null hypothesis given in
(\ref{2.2}).
\end{theorem}


The following theorem establishes the consistency of the Rao-type test
statistic of order $\beta$.

\begin{theorem}
Let $\boldsymbol{\theta\in\Theta}$ with $\boldsymbol{\theta\neq\theta}_{0}$
and we assume that $E_{\boldsymbol{\theta}}\left[  \boldsymbol{u}_{\beta
}\left(  X,\boldsymbol{\theta}_{0}\right)  \right]  \neq\boldsymbol{0}_{p}$.
Then,
\[
\lim_{n\rightarrow\infty}P\left(  R_{\beta,n}\left(  \boldsymbol{\theta}%
_{0}\right)  >\chi_{p,\alpha}^{2}\right)  =1.
\]

\end{theorem}

\begin{proof}
Because of the convergence %
\[
\boldsymbol{U}_{\beta,n}\left(  \boldsymbol{\theta}_{0}\right)
\underset{n\rightarrow\infty}{\overset{P}{\longrightarrow}}E\left[
\boldsymbol{u}_{\beta}\left(  X,\boldsymbol{\theta}_{0}\right)  \right]  ,
\]
we have%
\begin{align*}
P\left(  R_{\beta,n}\left(  \boldsymbol{\theta}_{0}\right)  >\chi_{p,\alpha
}^{2}\right)   &  =P\left(  \tfrac{1}{n}R_{\beta,n}\left(  \boldsymbol{\theta
}_{0}\right)  >\tfrac{1}{n}\chi_{p,\alpha}^{2}\right) \\
&  \underset{n\rightarrow\infty}{\longrightarrow}I\left(
E_{\boldsymbol{\theta}}\left[  \boldsymbol{u}_{\beta}\left(
X,\boldsymbol{\theta}_{0}\right)  \right]  \boldsymbol{K}_{\beta}^{-1}\left(
\boldsymbol{\theta}\right)  E_{\boldsymbol{\theta}}^{T}\left[  \boldsymbol{u}%
_{\beta}\left(  X,\boldsymbol{\theta}_{0}\right)  \right]  >0\right)  =1,
\end{align*}
where $I(\cdot)$ is an indicator function.
\end{proof}

Now we derive the asymptotic distribution of $R_{\beta,n}\left(
\boldsymbol{\theta}_{0}\right)  $ under local Pitman-type alternative
hypotheses of the form
\begin{equation}
H_{1,n}:\boldsymbol{\theta}=\boldsymbol{\theta}_{n}, \label{2.4}%
\end{equation}
where $\boldsymbol{\theta}_{n}=\boldsymbol{\theta}_{0}+n^{-1/2}\boldsymbol{d}%
$. Such results are helpful in determining the asymptotic contiguous power of
the Rao-type tests.

\begin{theorem}
\label{THM:cont_1}Under (\ref{2.4}), the asymptotic distribution of the
Rao-type test statistics $R_{\beta,n}\left(  \boldsymbol{\theta}_{0}\right)  $
is a non-central chi-square distribution with $p$ degrees of freedom and
non-centrality parameter given by
\begin{equation}
\delta_{\beta}(\boldsymbol{\theta}_{0},\boldsymbol{d})=\boldsymbol{d}%
^{T}\boldsymbol{J}_{\beta}\left(  \boldsymbol{\theta}_{0}\right)
\boldsymbol{K}_{\beta}^{-1}\left(  \boldsymbol{\theta}_{0}\right)
\boldsymbol{J}_{\beta}\left(  \boldsymbol{\theta}_{0}\right)  \boldsymbol{d}.
\label{cp}%
\end{equation}

\end{theorem}

\begin{proof}
Consider the Taylor series expansion%
\[
\sqrt{n}\boldsymbol{U}_{\beta,n}(\boldsymbol{\theta}_{n})=\sqrt{n}%
\boldsymbol{U}_{\beta,n}(\boldsymbol{\theta}_{0})+\left.  \frac{\partial
\boldsymbol{U}_{\beta,n}(\boldsymbol{\theta})}{\partial\boldsymbol{\theta}%
^{T}}\right\vert _{\boldsymbol{\theta=\theta}_{n}^{\ast}}\boldsymbol{d}%
\]
where $\boldsymbol{\theta}_{n}^{\ast}$ belongs to the line segment joining
$\boldsymbol{\theta}_{0}$ and $\boldsymbol{\theta}_{0}+\tfrac{1}{\sqrt{n}%
}\boldsymbol{d}$. By the CLT
\[
\sqrt{n}\boldsymbol{U}_{\beta,n}(\boldsymbol{\theta}_{0}%
)\underset{n\rightarrow\infty}{\overset{\mathcal{L}}{\longrightarrow}%
}\mathcal{N}\left(  \boldsymbol{0}_{p},\boldsymbol{K}_{\beta}\left(
\boldsymbol{\theta}_{0}\right)  \right)
\]
and by Khintchine Weak Law of Large Numbers and Slutsky's Theorem
\begin{equation}
\frac{\partial\boldsymbol{U}_{\beta,n}(\boldsymbol{\theta})}{\partial
\boldsymbol{\theta}^{T}}\underset{n\rightarrow\infty
}{\overset{P}{\longrightarrow}}E\left[  \frac{\partial\boldsymbol{u}_{\beta
}\left(  X,\boldsymbol{\theta}\right)  }{\partial\boldsymbol{\theta}^{T}%
}\right]  =-\boldsymbol{J}_{\beta}(\boldsymbol{\theta}). \label{minusJ}%
\end{equation}
The last equality arises from
\begin{align*}
\frac{\partial\boldsymbol{u}_{\beta}\left(  x,\boldsymbol{\theta}\right)
}{\partial\boldsymbol{\theta}^{T}}  &  =\frac{\partial\boldsymbol{s}%
_{\boldsymbol{\theta}}(x)}{\partial\boldsymbol{\theta}^{T}}%
f_{\boldsymbol{\theta}}^{\beta}(x)+\boldsymbol{s}_{\boldsymbol{\theta}%
}(x)f_{\boldsymbol{\theta}}^{\beta-1}(x)\beta\frac{\partial
f_{\boldsymbol{\theta}}(x)}{\partial\boldsymbol{\theta}^{T}}\\
&  -\int_{-\infty}^{+\infty}\frac{\partial\boldsymbol{s}_{\boldsymbol{\theta}%
}(y)}{\partial\boldsymbol{\theta}^{T}}f_{\boldsymbol{\theta}}^{\beta
+1}(y)dy-\int_{-\infty}^{+\infty}\boldsymbol{s}_{\boldsymbol{\theta}%
}(y)f_{\boldsymbol{\theta}}^{\beta}(y)\left(  \beta+1\right)  \frac{\partial
f_{\boldsymbol{\theta}}(y)}{\partial\boldsymbol{\theta}^{T}}dy
\end{align*}
and
\[
E\left[  \frac{\partial\boldsymbol{u}_{\beta}\left(  X,\boldsymbol{\theta
}\right)  }{\partial\boldsymbol{\theta}^{T}}\right]  =-\int_{-\infty}%
^{+\infty}\boldsymbol{s}_{\boldsymbol{\theta}}(x)\boldsymbol{s}%
_{\boldsymbol{\theta}}^{T}(x)f_{\boldsymbol{\theta}}^{\beta+1}%
(x)dx=-\boldsymbol{J}_{\beta}\left(  \boldsymbol{\theta}\right)  .
\]
Therefore,
\[
\left.  \sqrt{n}\boldsymbol{U}_{\beta,n}\left(  \boldsymbol{\theta}\right)
\right\vert _{\boldsymbol{\theta=\theta}_{0}\boldsymbol{+}n^{-1/2}%
\boldsymbol{d}}\underset{n\rightarrow\infty}{\overset{\mathcal{L}%
}{\longrightarrow}}\mathcal{N}\left(  -\boldsymbol{J}_{\beta}\left(
\boldsymbol{\theta}_{0}\right)  \boldsymbol{d},\boldsymbol{K}_{\beta}\left(
\boldsymbol{\theta}_{0}\right)  \right)
\]
and
\[
R_{\beta,n}\left(  \boldsymbol{\theta}_{0}\right)  \underset{n\rightarrow
\infty}{\overset{\mathcal{L}}{\longrightarrow}}\chi_{p}^{2}\left(
\delta_{\beta}(\boldsymbol{\theta}_{0},\boldsymbol{d})\right)  ,
\]
with $\delta_{\beta}(\boldsymbol{\theta}_{0},\boldsymbol{d})$ given by
(\ref{cp}).
\end{proof}


\section{Rao-type test statistics for composite null hypotheses}


\label{SEC:composite}

Based on the RMDPDE, we define a Rao-type test statistics for testing (\ref{1}).

\begin{definition}
The Rao-type test statistics of order $\beta$ for testing (\ref{1}) based on
the RMDPDE, $\boldsymbol{\tilde{\theta}}_{\beta}$, is given by
\begin{equation}
\widetilde{R}_{\beta,n}(\boldsymbol{\tilde{\theta}}_{\beta})=n\boldsymbol{U}%
_{\beta,n}^{T}(\boldsymbol{\tilde{\theta}}_{\beta})\boldsymbol{Q}_{\beta
}(\boldsymbol{\tilde{\theta}}_{\beta})\left[  \boldsymbol{Q}_{\beta}%
^{T}(\boldsymbol{\tilde{\theta}}_{\beta})\boldsymbol{K}_{\beta}%
(\boldsymbol{\tilde{\theta}}_{\beta})\boldsymbol{Q}_{\beta}(\boldsymbol{\tilde
{\theta}}_{\beta})\right]  ^{-1}\boldsymbol{Q}_{\beta}^{T}(\boldsymbol{\tilde
{\theta}}_{\beta})\boldsymbol{U}_{\beta,n}(\boldsymbol{\tilde{\theta}}_{\beta
}), \label{RaC}%
\end{equation}
with $\boldsymbol{U}_{\beta,n}(\boldsymbol{\theta})$ as defined in Equation
(\ref{2.1}) and $\boldsymbol{Q}_{\beta}(\boldsymbol{\theta})$ in (\ref{Q}).
\end{definition}

We will now derive the asymptotic distribution of $\widetilde{R}_{\beta
,n}(\boldsymbol{\tilde{\theta}}_{\beta})$.

\begin{theorem}
Let $X_{1},\ldots,X_{n}$ be i.i.d. random variables with density function
$f_{\boldsymbol{\theta}}(x)$, $\boldsymbol{\theta}\in\Theta\subset
\mathbb{R}^{p}$ satisfying some regularity conditions and consider the problem
of testing (\ref{1}). The Rao-type test statistics $\widetilde{R}_{\beta
,n}(\boldsymbol{\tilde{\theta}}_{\beta})$ has an asymptotic chi-square
distribution with $r$ degrees of freedom under $H_{0}$ given in (\ref{1}).
\end{theorem}

\begin{proof}
The RMDPDE, $\boldsymbol{\tilde{\theta}}_{\beta}$, must satisfy the restricted
equations on $\boldsymbol{\theta}$
\begin{equation}
\boldsymbol{\tilde{U}}_{\beta,n}(\boldsymbol{\theta})=\boldsymbol{0}_{p}
\label{E}%
\end{equation}
and $\boldsymbol{m}(\boldsymbol{\theta})=\boldsymbol{0}_{r}$, where%
\begin{equation}
\boldsymbol{\tilde{U}}_{\beta,n}(\boldsymbol{\theta})=\boldsymbol{U}_{\beta
,n}(\boldsymbol{\theta})+\boldsymbol{M}(\boldsymbol{\theta})\boldsymbol{\tilde
{\lambda}}_{\beta,n}, \label{RestrScore}%
\end{equation}
with $\boldsymbol{\tilde{\lambda}}_{\beta,n}=\boldsymbol{\tilde{\lambda}%
}_{\beta,n}(X_{1},...,X_{n},\boldsymbol{\theta})\in\mathbb{R}^{r}$
being the vector of Lagrangian multipliers associated to $\beta\in%
\mathbb{R}
^{+}$ and $\boldsymbol{U}_{\beta,n}(\boldsymbol{\theta})$\ was given in
(\ref{2.1}).
We consider the Taylor expansion of $\boldsymbol{U}_{\beta,n}%
(\boldsymbol{\tilde{\theta}}_{\beta})$ at the point $\boldsymbol{\theta}_{0}%
$,
\[
\boldsymbol{U}_{\beta,n}(\boldsymbol{\tilde{\theta}}_{\beta})=\boldsymbol{U}%
_{\beta,n}(\boldsymbol{\theta}_{0})+\left.  \frac{\partial}{\partial
\boldsymbol{\theta}}\boldsymbol{U}_{\beta,n}^{T}(\boldsymbol{\theta
})\right\vert _{\boldsymbol{\theta}=\boldsymbol{\theta}_{0}}%
(\boldsymbol{\tilde{\theta}}_{\beta}-\boldsymbol{\theta}_{0})+o\left(
||\boldsymbol{\tilde{\theta}}_{\beta}-\boldsymbol{\theta}_{0}||^{2}%
\boldsymbol{1}_{p}\right)  .
\]
Therefore,
\[
\sqrt{n}\boldsymbol{U}_{\beta,n}(\boldsymbol{\tilde{\theta}}_{\beta})=\sqrt
{n}\boldsymbol{U}_{\beta,n}(\boldsymbol{\theta}_{0})+\sqrt{n}\left.
\frac{\partial}{\partial\boldsymbol{\theta}}\boldsymbol{U}_{\beta,n}%
^{T}(\boldsymbol{\theta})\right\vert _{\boldsymbol{\theta}=\boldsymbol{\theta
}_{0}}(\boldsymbol{\tilde{\theta}}_{\beta}-\boldsymbol{\theta}_{0})+o\left(
\sqrt{n}||\boldsymbol{\tilde{\theta}}_{\beta}-\boldsymbol{\theta}_{0}%
||^{2}\boldsymbol{1}_{p}\right)  .
\]
By (\ref{minusJ}) it holds,%
\[
\sqrt{n}\boldsymbol{U}_{\beta,n}(\boldsymbol{\tilde{\theta}}_{\beta})=\sqrt
{n}\boldsymbol{U}_{\beta,n}(\boldsymbol{\theta}_{0})-\sqrt{n}\boldsymbol{J}%
_{\beta}(\boldsymbol{\theta}_{0})(\boldsymbol{\tilde{\theta}}_{\beta
}-\boldsymbol{\theta}_{0})+o\left(  \sqrt{n}||\boldsymbol{\tilde{\theta}%
}_{\beta}-\boldsymbol{\theta}_{0}||^{2}\boldsymbol{1}_{p}\right)  +o_{P}(\boldsymbol{1}_{p}).
\]
A Taylor expansion of $\boldsymbol{m}(\boldsymbol{\tilde{\theta}}_{\beta})$
around the point $\boldsymbol{\theta}_{0}$ gives
\begin{equation}
\sqrt{n}\boldsymbol{m}(\boldsymbol{\tilde{\theta}}_{\beta})=\boldsymbol{M}%
^{T}(\boldsymbol{\theta}_{0})\sqrt{n}(\boldsymbol{\tilde{\theta}}_{\beta
}-\boldsymbol{\theta}_{0})+o_{P}(\boldsymbol{1}_{r}). \label{13}%
\end{equation}
By (\ref{E}), we have
\begin{equation}
\sqrt{n}\boldsymbol{U}_{\beta,n}(\boldsymbol{\theta}_{0})-\sqrt{n}%
\boldsymbol{J}_{\beta}(\boldsymbol{\theta}_{0})(\boldsymbol{\tilde{\theta}%
}_{\beta}-\boldsymbol{\theta}_{0})+\boldsymbol{M}(\boldsymbol{\theta}%
_{0})\sqrt{n}\boldsymbol{\tilde{\lambda}}_{\beta,n}=o_{P}(\boldsymbol{1}_{p}),
\label{14}%
\end{equation}
and by (\ref{13})
\begin{equation}
\boldsymbol{M}^{T}(\boldsymbol{\theta}_{0})\sqrt{n}(\boldsymbol{\tilde{\theta
}}_{\beta}-\boldsymbol{\theta}_{0})+o_{P}(1)=\boldsymbol{0}_{r}. \label{15}%
\end{equation}
Now we are going to write (\ref{14}) and (\ref{15}) in a matrix form as%
\[%
\begin{pmatrix}
-\boldsymbol{J}_{\beta}(\boldsymbol{\theta}_{0}) & \boldsymbol{M}%
(\boldsymbol{\theta}_{0})\\
\boldsymbol{M}^{T}(\boldsymbol{\theta}_{0}) & \boldsymbol{O}_{r\times r}%
\end{pmatrix}%
\begin{pmatrix}
\sqrt{n}(\boldsymbol{\tilde{\theta}}_{\beta}-\boldsymbol{\theta}_{0})\\
\sqrt{n}\boldsymbol{\tilde{\lambda}}_{\beta,n}%
\end{pmatrix}
=%
\begin{pmatrix}
-\sqrt{n}\boldsymbol{U}_{\beta,n}(\boldsymbol{\theta}_{0})\\
\boldsymbol{0}_{r}%
\end{pmatrix}
+o_{P}(\boldsymbol{1}_{p+r}).
\]
Therefore,%
\[%
\begin{pmatrix}
\sqrt{n}(\boldsymbol{\tilde{\theta}}_{\beta}-\boldsymbol{\theta}_{0})\\
\sqrt{n}\boldsymbol{\tilde{\lambda}}_{\beta,n}%
\end{pmatrix}
=%
\begin{pmatrix}
-\boldsymbol{J}_{\beta}(\boldsymbol{\theta}_{0}) & \boldsymbol{M}%
(\boldsymbol{\theta}_{0})\\
\boldsymbol{M}^{T}(\boldsymbol{\theta}_{0}) & \boldsymbol{O}_{r\times r}%
\end{pmatrix}
^{-1}%
\begin{pmatrix}
-\sqrt{n}\boldsymbol{U}_{\beta,n}(\boldsymbol{\theta}_{0})\\
\boldsymbol{0}_{r}%
\end{pmatrix}
+o_{P}(\boldsymbol{1}_{p+r}),
\]
i.e.,%
\[%
\begin{pmatrix}
\sqrt{n}(\boldsymbol{\tilde{\theta}}_{\beta}-\boldsymbol{\theta}_{0})\\
\sqrt{n}\boldsymbol{\tilde{\lambda}}_{\beta,n}%
\end{pmatrix}
=%
\begin{pmatrix}
\boldsymbol{P}_{\beta}(\boldsymbol{\theta}_{0}) & \boldsymbol{Q}_{\beta
}(\boldsymbol{\theta}_{0})\\
\boldsymbol{Q}_{\beta}^{T}(\boldsymbol{\theta}_{0}) & \boldsymbol{R}_{\beta
}(\boldsymbol{\theta}_{0})
\end{pmatrix}%
\begin{pmatrix}
-\sqrt{n}\boldsymbol{U}_{\beta,n}(\boldsymbol{\theta}_{0})\\
\boldsymbol{0}_{r}%
\end{pmatrix}
+o_{P}(\boldsymbol{1}_{p+r}),
\]
with $\boldsymbol{R}_{\beta}(\boldsymbol{\theta}_{0})=\left(  \boldsymbol{M}%
_{\beta}^{T}(\boldsymbol{\tilde{\theta}})\boldsymbol{J}_{\beta}^{-1}%
(\boldsymbol{\tilde{\theta}})\boldsymbol{M}_{\beta}(\boldsymbol{\tilde{\theta
}})\right)  ^{-1}$.\ The matrices $\boldsymbol{P}_{\beta}(\boldsymbol{\theta
}_{0})$ and $\boldsymbol{Q}_{\beta}(\boldsymbol{\theta}_{0})$ were defined in
(\ref{P}) and (\ref{Q}). But,
\[%
\begin{pmatrix}
\sqrt{n}\boldsymbol{U}_{\beta,n}(\boldsymbol{\theta}_{0})\\
\boldsymbol{0}_{r}%
\end{pmatrix}
\underset{n\rightarrow\infty}{\overset{\mathcal{L}}{\longrightarrow}%
}\mathcal{N}\left(  \boldsymbol{0}_{p+r},%
\begin{pmatrix}
\boldsymbol{K}_{\beta}({\boldsymbol{\theta}}_{0}) & \boldsymbol{O}_{p\times
p}\\
\boldsymbol{O}_{r\times r} & \boldsymbol{O}_{r\times p}%
\end{pmatrix}
\right)  .
\]
Thus,
\[%
\begin{pmatrix}
\sqrt{n}(\boldsymbol{\tilde{\theta}}_{\beta}-\boldsymbol{\theta}_{0})\\
\sqrt{n}\boldsymbol{\tilde{\lambda}}_{\beta,n}%
\end{pmatrix}
\underset{n\rightarrow\infty}{\overset{\mathcal{L}}{\longrightarrow}%
}\mathcal{N}\left(  \boldsymbol{0}_{p+r},\boldsymbol{\Sigma}_{\beta
}({\boldsymbol{\theta}}_{0})\right)  ,
\]
where
\begin{align}
\boldsymbol{\Sigma}_{\beta}({\boldsymbol{\theta}}_{0})  &  =%
\begin{pmatrix}
\boldsymbol{P}_{\beta}(\boldsymbol{\theta}_{0}) & \boldsymbol{Q}_{\beta
}(\boldsymbol{\theta}_{0})\\
\boldsymbol{Q}_{\beta}^{T}(\boldsymbol{\theta}_{0}) & \boldsymbol{R}_{\beta
}(\boldsymbol{\theta}_{0})
\end{pmatrix}%
\begin{pmatrix}
\boldsymbol{K}_{\beta}({\boldsymbol{\theta}}_{0}) & \boldsymbol{O}_{p\times
p}\\
\boldsymbol{O}_{r\times r} & \boldsymbol{O}_{r\times p}%
\end{pmatrix}%
\begin{pmatrix}
\boldsymbol{P}_{\beta}^{T}(\boldsymbol{\theta}_{0}) & \boldsymbol{Q}_{\beta
}(\boldsymbol{\theta}_{0})\\
\boldsymbol{Q}_{\beta}^{T}(\boldsymbol{\theta}_{0}) & \boldsymbol{R}_{\beta
}^{T}(\boldsymbol{\theta}_{0})
\end{pmatrix}
\nonumber\\
&  =%
\begin{pmatrix}
\boldsymbol{P}_{\beta}(\boldsymbol{\theta}_{0})\boldsymbol{K}_{\beta
}(\boldsymbol{\theta}_{0})\boldsymbol{P}_{\beta}^{T}(\boldsymbol{\theta}%
_{0}) & \boldsymbol{P}_{\beta}(\boldsymbol{\theta}_{0})\boldsymbol{K}_{\beta
}(\boldsymbol{\theta}_{0})\boldsymbol{Q}_{\beta}(\boldsymbol{\theta}_{0})\\
\boldsymbol{Q}_{\beta}^{T}(\boldsymbol{\theta}_{0})\boldsymbol{K}_{\beta
}(\boldsymbol{\theta}_{0})\boldsymbol{P}_{\beta}^{T}(\boldsymbol{\theta}%
_{0}) & \boldsymbol{Q}_{\beta}^{T}(\boldsymbol{\theta}_{0})\boldsymbol{K}%
_{\beta}(\boldsymbol{\theta}_{0})\boldsymbol{Q}_{\beta}(\boldsymbol{\theta
}_{0})
\end{pmatrix}
.\nonumber
\end{align}
Then,
\begin{equation}
\sqrt{n}\boldsymbol{\tilde{\lambda}}_{\beta,n}\underset{n\rightarrow
\infty}{\overset{\mathcal{L}}{\longrightarrow}}\mathcal{N}\left(
\boldsymbol{0}_{r},\boldsymbol{Q}_{\beta}^{T}(\boldsymbol{\theta}%
_{0})\boldsymbol{K}_{\beta}(\boldsymbol{\theta}_{0})\boldsymbol{Q}_{\beta
}(\boldsymbol{\theta}_{0})\right)  . \label{Ni3}%
\end{equation}
From (\ref{E}), we know that $\boldsymbol{U}_{\beta,n}(\boldsymbol{\tilde
{\theta}}_{\beta})=-\boldsymbol{M}(\boldsymbol{\tilde{\theta}}_{\beta
})\boldsymbol{\tilde{\lambda}}_{\beta,n}$,
and hence
\begin{align*}
\boldsymbol{U}_{\beta,n}^{T}(\boldsymbol{\tilde{\theta}}_{\beta}%
)\boldsymbol{Q}_{\beta}(\boldsymbol{\tilde{\theta}}_{\beta})  &
=-\boldsymbol{\tilde{\lambda}}_{\beta,n}^{T}\boldsymbol{M}^{T}%
(\boldsymbol{\tilde{\theta}}_{\beta})\boldsymbol{Q}_{\beta}(\boldsymbol{\tilde
{\theta}}_{\beta})\\
&  =-\boldsymbol{\tilde{\lambda}}_{\beta,n}^{T}\boldsymbol{M}^{T}%
(\boldsymbol{\tilde{\theta}}_{\beta})\boldsymbol{J}_{\beta}^{-1}%
(\boldsymbol{\tilde{\theta}}_{\beta})\boldsymbol{M}(\boldsymbol{\tilde{\theta
}}_{\beta})\left[  \boldsymbol{M}(\boldsymbol{\tilde{\theta}}_{\beta
})\boldsymbol{J}_{\beta}^{-1}(\boldsymbol{\tilde{\theta}}_{\beta
})\boldsymbol{M}(\boldsymbol{\tilde{\theta}}_{\beta})\right]  ^{-1}\\
&  =-\boldsymbol{\tilde{\lambda}}_{\beta,n}^{T}.
\end{align*}
Therefore%
\begin{equation}
\widetilde{R}_{\beta,n}(\boldsymbol{\tilde{\theta}}_{\beta}%
)=n\boldsymbol{\tilde{\lambda}}_{\beta,n}^{T}\left[  \boldsymbol{Q}_{\beta
}^{T}(\boldsymbol{\tilde{\theta}}_{\beta})\boldsymbol{K}_{\beta}%
(\boldsymbol{\tilde{\theta}}_{\beta})\boldsymbol{Q}_{\beta}(\boldsymbol{\tilde
{\theta}}_{\beta})\right]  ^{-1}\boldsymbol{\tilde{\lambda}}_{\beta,n}
\label{EQ:LM_test_Gen}%
\end{equation}
and now the result follows by (\ref{Ni3}).\bigskip
\end{proof}

It is easily seen that the RMDPDE coincides with the restricted maximum
likelihood estimator of $\boldsymbol{\theta}$ (RMLE), $\boldsymbol{\tilde
{\theta}}$, for $\beta=0$. In the next proposition it is proved that we
recover the classical Rao test statistic (see Rao, 2005) in such a case.

\begin{proposition}
For $\beta=0$, the Rao-type test statistic in Equation (\ref{RaC}),
$\widetilde{R}_{\beta=0,n}(\boldsymbol{\tilde{\theta}}_{\beta})$, coincides
with the classical Rao test statistic $\widetilde{R}_{n}(\boldsymbol{\tilde
{\theta}})$.
\end{proposition}

\begin{proof}
It is immediately seen that, at the model for $\beta=0$, we have
$\boldsymbol{K}_{\beta=0}(\boldsymbol{\theta})=\boldsymbol{J}_{\beta
=0}(\boldsymbol{\theta})=\boldsymbol{I}(\boldsymbol{\theta})$, with
$\boldsymbol{I}(\boldsymbol{\theta})$ being the Fisher information matrix
associated to the model; see Equations (\ref{model_variance}) and (\ref{m_v}).
Therefore, at $\beta=0$, we have
\[
\boldsymbol{Q}_{\beta=0}(\boldsymbol{\theta})=\boldsymbol{I}^{-1}%
(\boldsymbol{\theta})\boldsymbol{M}(\boldsymbol{\theta})\left[  \boldsymbol{M}%
^{T}(\boldsymbol{\theta})\boldsymbol{I}^{-1}(\boldsymbol{\theta}%
)\boldsymbol{M}(\boldsymbol{\theta})\right]  ^{-1}%
\]
and
\begin{align*}
&  \boldsymbol{Q}_{\beta=0}^{T}(\boldsymbol{\theta})\boldsymbol{K}_{\beta
=0}(\boldsymbol{\theta})\boldsymbol{Q}_{\beta=0}(\boldsymbol{\theta})\\
&  =\left[  \boldsymbol{M}^{T}(\boldsymbol{\theta})\boldsymbol{I}%
^{-1}(\boldsymbol{\theta})\boldsymbol{M}(\boldsymbol{\theta})\right]
^{-1}\boldsymbol{M}^{T}(\boldsymbol{\theta})\boldsymbol{I}^{-1}%
(\boldsymbol{\theta})\boldsymbol{I}(\boldsymbol{\theta})\boldsymbol{I}%
^{-1}(\boldsymbol{\theta})\boldsymbol{M}(\boldsymbol{\theta})\left[
\boldsymbol{M}^{T}(\boldsymbol{\theta})\boldsymbol{I}^{-1}(\boldsymbol{\theta
})\boldsymbol{M}(\boldsymbol{\theta})\right]  ^{-1}\\
&  =\left[  \boldsymbol{M}^{T}(\boldsymbol{\theta})\boldsymbol{I}%
^{-1}(\boldsymbol{\theta})\boldsymbol{M}(\boldsymbol{\theta})\right]  ^{-1}.
\end{align*}
Therefore,
\begin{align}
&  \boldsymbol{Q}_{\beta=0}(\boldsymbol{\theta})\left[  \boldsymbol{Q}%
_{\beta=0}^{T}(\boldsymbol{\theta})\boldsymbol{K}_{\beta=0}(\boldsymbol{\theta
})\boldsymbol{Q}_{\beta=0}(\boldsymbol{\theta})\right]  ^{-1}\boldsymbol{Q}%
_{\beta=0}^{T}(\boldsymbol{\theta})\nonumber\\
&  =\boldsymbol{I}^{-1}(\boldsymbol{\theta})\boldsymbol{M}(\boldsymbol{\theta
})\left[  \boldsymbol{M}^{T}(\boldsymbol{\theta})\boldsymbol{I}^{-1}%
(\boldsymbol{\theta})\boldsymbol{M}(\boldsymbol{\theta})\right]  ^{-1}\left[
\boldsymbol{M}^{T}(\boldsymbol{\theta})\boldsymbol{I}^{-1}(\boldsymbol{\theta
})\boldsymbol{M}(\boldsymbol{\theta})\right]  \left[  \boldsymbol{M}%
^{T}(\boldsymbol{\theta})\boldsymbol{I}^{-1}(\boldsymbol{\theta}%
)\boldsymbol{M}(\boldsymbol{\theta})\right]  ^{-1}\nonumber\\
&  \times\boldsymbol{M}^{T}(\boldsymbol{\theta})\boldsymbol{I}^{-1}%
(\boldsymbol{\theta})\nonumber\\
&  =\boldsymbol{I}^{-1}(\boldsymbol{\theta})\boldsymbol{M}(\boldsymbol{\theta
})\left[  \boldsymbol{M}^{T}(\boldsymbol{\theta})\boldsymbol{I}^{-1}%
(\boldsymbol{\theta})\boldsymbol{M}(\boldsymbol{\theta})\right]
^{-1}\boldsymbol{M}^{T}(\boldsymbol{\theta})\boldsymbol{I}^{-1}%
(\boldsymbol{\theta}). \label{8}%
\end{align}
Since $\boldsymbol{\tilde{\theta}}_{\beta=0}=\boldsymbol{\tilde{\theta}}$ and
$\boldsymbol{U}$$_{\beta=0,n}(\boldsymbol{\theta})=\boldsymbol{U}$%
$_{n}(\boldsymbol{\theta})$, the Rao-type test statistic in (\ref{RaC})
simplifies to
\begin{align}
\widetilde{R}_{\beta=0,n}(\boldsymbol{\tilde{\theta}})  &  =n\boldsymbol{U}%
_{\beta=0,n}^{T}(\boldsymbol{\tilde{\theta}}_{0})\boldsymbol{Q}_{\beta
=0}(\boldsymbol{\tilde{\theta}}_{\beta=0})\left[  \boldsymbol{Q}_{\beta=0}%
^{T}(\boldsymbol{\tilde{\theta}}_{\beta=0})\boldsymbol{K}_{\beta
=0}(\boldsymbol{\tilde{\theta}}_{\beta=0})\boldsymbol{Q}_{\beta=0}%
(\boldsymbol{\tilde{\theta}}_{\beta=0})\right]  ^{-1}\nonumber\\
&  \times\boldsymbol{Q}_{\beta=0}^{T}(\boldsymbol{\tilde{\theta}}_{\beta
=0})\boldsymbol{U}_{\beta=0,n}(\boldsymbol{\tilde{\theta}}_{\beta
=0})\nonumber\\
&  =n\boldsymbol{U}_{n}^{T}(\boldsymbol{\tilde{\theta}})\boldsymbol{I}%
^{-1}(\boldsymbol{\tilde{\theta}})\boldsymbol{M}(\boldsymbol{\tilde{\theta}%
})\left[  \boldsymbol{M}^{T}(\boldsymbol{\tilde{\theta}})\boldsymbol{I}%
^{-1}(\boldsymbol{\tilde{\theta}})\boldsymbol{M}(\boldsymbol{\tilde{\theta}%
})\right]  ^{-1}\boldsymbol{M}^{T}(\boldsymbol{\tilde{\theta}})\boldsymbol{I}%
^{-1}(\boldsymbol{\tilde{\theta}})\boldsymbol{U}_{n}(\boldsymbol{\tilde
{\theta}}), \label{EQ:Rao_C_0}%
\end{align}
which is the classical Rao test statistic for general composite hypothesis
given in (\ref{1}).

\end{proof}

\begin{remark}
\normalfont Let us consider the most common composite hypothesis related to
the problem of testing a part of the parameter vector. Consider the partition
$\boldsymbol{\theta}=(\boldsymbol{\theta}_{1}^{T},\boldsymbol{\theta}_{2}%
^{T})^{T}$, with $\boldsymbol{\theta}_{1}$ denoting the first $r$ components
of the parameter vector, and the hypothesis $H_{0}:\boldsymbol{\theta}%
_{1}=\boldsymbol{\theta}_{10}$ for some pre-fixed $r$-vector
$\boldsymbol{\theta}_{10}$ against the omnibus alternative. Note that this
case belongs to the general set-up of hypothesis in (\ref{1}) with
$\boldsymbol{m}(\boldsymbol{\theta})=\boldsymbol{\theta}_{1}%
-\boldsymbol{\theta}_{10}$ and $\boldsymbol{M}^{T}(\boldsymbol{\theta
})=[\boldsymbol{I}_{r}~~\boldsymbol{O}_{r\times p}]$. For this particular
problem, we can easily simplify the proposed Rao-type test statistics from
(\ref{RaC}) to%
\begin{equation}
\widetilde{R}_{\beta,n}(\boldsymbol{\tilde{\theta}}_{\beta})=n\boldsymbol{U}%
_{\beta,n,1}^{T}(\boldsymbol{\tilde{\theta}}_{\beta})\boldsymbol{K}_{\beta
,11}^{-1}(\boldsymbol{\tilde{\theta}}_{\beta})\boldsymbol{U}_{\beta
,n,1}(\boldsymbol{\tilde{\theta}}_{\beta}), \label{redR}%
\end{equation}
where $\boldsymbol{U}_{\beta,n,1}(\boldsymbol{\theta})$ denotes the first $r$
components of $\boldsymbol{U}_{\beta,n}(\boldsymbol{\theta})$ and
$\boldsymbol{K}_{\beta,11}(\boldsymbol{\theta})$ represents the $r\times r$
principle sub-matrix of $\boldsymbol{K}_{\beta}(\boldsymbol{\theta})$. The
interesting case is $\beta=0$, where $\widetilde{R}_{\beta=0,n}%
(\boldsymbol{\tilde{\theta}})$ is simplified to
\[
\widetilde{R}_{n}(\boldsymbol{\tilde{\theta}})=n\boldsymbol{U}_{n,1}%
^{T}(\boldsymbol{\tilde{\theta}})\left(  \boldsymbol{I}_{11}%
(\boldsymbol{\tilde{\theta}})-\boldsymbol{I}_{12}(\boldsymbol{\tilde{\theta}%
})\boldsymbol{I}_{22}^{-1}(\boldsymbol{\tilde{\theta}})\boldsymbol{I}%
_{21}(\boldsymbol{\tilde{\theta}})\right)  ^{-1}\boldsymbol{U}_{n,1}%
(\boldsymbol{\tilde{\theta}}),
\]
where $\boldsymbol{U}_{n,1}(\boldsymbol{\theta})=\boldsymbol{U}_{\beta
=0,n,1}(\boldsymbol{\theta})$, $\left(  \boldsymbol{I}_{11}(\boldsymbol{\theta
})-\boldsymbol{I}_{12}(\boldsymbol{\theta})\boldsymbol{I}_{22}^{-1}%
(\boldsymbol{\theta})\boldsymbol{I}_{21}(\boldsymbol{\theta})\right)  ^{-1}$
represents the $r\times r$ principle sub-matrix of $\boldsymbol{I}%
^{-1}(\boldsymbol{\theta})$, with%
\[
\boldsymbol{I}(\boldsymbol{\theta})=%
\begin{pmatrix}
\boldsymbol{I}_{11}(\boldsymbol{\theta}) & \boldsymbol{I}_{12}%
(\boldsymbol{\theta})\\
\boldsymbol{I}_{21}(\boldsymbol{\theta}) & \boldsymbol{I}_{22}%
(\boldsymbol{\theta})
\end{pmatrix}
\]
being the block structure of the Fisher Information matrix. This is exactly of
the same form as given in Rao's original paper (Rao, 1973) on this particular
testing problem.
\end{remark}

\begin{remark}
\normalfont Note that, in view of (\ref{E}) applied at $\beta=0$, we have
$\boldsymbol{U}$$_{n}^{0}(\boldsymbol{\tilde{\theta}}_{\beta=0}%
)=-\boldsymbol{M}(\boldsymbol{\tilde{\theta}}_{\beta=0})\boldsymbol{\tilde
{\lambda}}_{n}^{\beta=0}$, and hence either from (\ref{EQ:Rao_C_0}) or
(\ref{EQ:LM_test_Gen}) we can rewrite the Rao-type test statistics at
$\beta=0$ as follows
\begin{equation}
\widetilde{R}_{\beta=0,n}(\boldsymbol{\tilde{\theta}})=n\boldsymbol{\tilde
{\lambda}}_{n}^{T}\boldsymbol{M}^{T}(\boldsymbol{\tilde{\theta}}%
)\boldsymbol{I}^{-1}(\boldsymbol{\tilde{\theta}})\boldsymbol{M}%
(\boldsymbol{\tilde{\theta}})\boldsymbol{\tilde{\lambda}}_{n},
\label{EQ:LM_test_0}%
\end{equation}
where $\boldsymbol{\tilde{\lambda}}_{n}=\boldsymbol{\tilde{\lambda}}%
_{\beta=0,n}$ is the Lagrange multiplier corresponding to the restricted MLE.
The tests statistic in (\ref{EQ:LM_test_0}) is another popular form of the
Rao test which is often referred to as the \textit{Lagrange Multiplier Test}
in econometrics. Hence, we can also see the proposed Rao-type test statistic
in (\ref{RaC}) as a \textit{Generalized Lagrange Multiplier Test} as well.
\end{remark}



\section{Robustness Analysis}



\subsection{Influence Function of Rao-type test statistics}


The influence function is a classical tool to measure infinitesimal robustness
of any general statistic. Let us first study the influence function of the
proposed Rao-type test statistics to examine their robustness against data
contamination. For this purpose, we need to redefine the Rao-type test
statistics in terms of a statistical functional.

The Rao-type test statistics for the simple null hypothesis is defined in
terms of the MDPDE estimating equations. The statistical functional, say
$\boldsymbol{T}_{\beta,G}$ at the true distribution $G$, associated with the
MDPDE is defined as the minimizer of the DPD measure between true density $g$
of $G$ and the model density $f_{\boldsymbol{\theta}}$, or equivalently as the
solution in ${{\boldsymbol{\theta}}}\in\Theta$ of $\boldsymbol{U}_{\beta
,G}\left(  \boldsymbol{\theta}\right)  =\boldsymbol{0}_{p}$ where%
\begin{equation}
\boldsymbol{U}_{\beta,G}\left(  \boldsymbol{\theta}\right)  =\int_{-\infty
}^{+\infty}\boldsymbol{u}_{\beta}\left(  x,\boldsymbol{\theta}\right)
dG(x)=\left(  \int_{-\infty}^{+\infty}u_{1,\beta}\left(  x,\boldsymbol{\theta
}\right)  dG(x),....,\int_{-\infty}^{+\infty}u_{p,\beta}\left(
x,\boldsymbol{\theta}\right)  dG(x)\right)  ^{T}, \label{EQ:FunctU}%
\end{equation}
where $\boldsymbol{u}_{\beta}\left(  x,\boldsymbol{\theta}\right)  $\ is given
in (\ref{EQ:score_function}). So, the functional corresponding to
$\boldsymbol{U}_{\beta,n}\left(  \boldsymbol{\theta}\right)  $ in (\ref{2.1})
can be defined as $\boldsymbol{U}_{\beta,G}\left(  \boldsymbol{\theta}\right)
$ and hence (ignoring the multiplier $n$) the statistical functionals
associated with the proposed Rao-type test statistics $R_{\beta,n}\left(
\boldsymbol{\theta}_{0}\right)  $ for testing the simple null hypothesis
(\ref{2.2}) are given by
\begin{equation}
R_{\beta,G}\left(  \boldsymbol{\theta}_{0}\right)  =\boldsymbol{U}_{\beta
,G}^{T}\left(  \boldsymbol{\theta}_{0}\right)  \boldsymbol{K}_{\beta}%
^{-1}\left(  \boldsymbol{\theta}_{0}\right)  \boldsymbol{U}_{\beta,G}\left(
\boldsymbol{\theta}_{0}\right)  . \label{EQ:SimpleTest _Func}%
\end{equation}
Note that, at the null hypothesis $G=F_{\boldsymbol{\theta}_{0}}$, we have
$\boldsymbol{T}_{\beta}(F_{\boldsymbol{\theta}_{0}})=\boldsymbol{\theta}_{0}$
(by Fisher consistency of the MDPDE) and $\boldsymbol{U}%
_{F_{\boldsymbol{\theta}_{0}}}^{\beta}(\boldsymbol{\theta}_{0})=\boldsymbol{0}%
_{p}$ so that $R_{F_{\boldsymbol{\theta}_{0}}}^{\beta}\left(
\boldsymbol{\theta}_{0}\right)  =0$.

In case of testing composite hypothesis, the corresponding Rao-type test
statistics are defined in terms of the RMDPDE. The statistical function and
influence function of the estimators under parametric restrictions have been
rigorously studied in Ghosh (2015). Following this approach, the statistical
functional associated with the MDPDE at the true distribution $G$, say
$\boldsymbol{\tilde{T}}_{\beta,G}$, is defined as the minimizer of $d_{\beta
}(g,f_{\boldsymbol{\theta}})$ subject to the null restrictions $\boldsymbol{m}%
(\boldsymbol{\theta)=0}_{r}$; the estimating equations can be written in terms
of Lagrange multipliers as in Section 4. However, for the influence function
analysis, we adopt the alternative approach of Ghosh (2015). Then
$\boldsymbol{\tilde{T}}_{\beta,G}$ can be thought of as a solution in
${{\boldsymbol{\theta}}}\in\Theta_{0}\subset\Theta$ of $\boldsymbol{U}%
_{\beta,G}(\boldsymbol{\theta})=\boldsymbol{0}_{p}$,
and its existence can be verified rigorously through the Implicit Function
Theorem. Following Ghosh (2015) its influence function is
\begin{equation}
\mathcal{IF}(y,\boldsymbol{\tilde{T}}_{\beta,G})=\boldsymbol{J}_{\beta}%
^{-1}({{\boldsymbol{\theta}}})\left[  \boldsymbol{u}_{\beta}%
(y,\boldsymbol{\tilde{T}}_{\beta,G})-\boldsymbol{U}_{\beta,G}%
(\boldsymbol{\tilde{T}}_{\beta,G})\right]  . \label{IFEst}%
\end{equation}
Now, (ignoring the multiplier $n$) we define the statistical functionals
associated with the proposed Rao-type test statistics $\tilde{R}_{\beta
,n}(\boldsymbol{\tilde{\theta}})$ for testing the composite null hypothesis
(\ref{2.2}) as given by
\begin{equation}
\tilde{R}_{\beta,G}(\boldsymbol{\tilde{T}}_{\beta,G})=\boldsymbol{U}_{\beta
,G}^{T}(\boldsymbol{\tilde{T}}_{\beta,G})\boldsymbol{Q}_{\beta}%
(\boldsymbol{\theta})\left[  \boldsymbol{Q}_{\beta}^{T}(\boldsymbol{\theta
})\boldsymbol{K}_{\beta}(\boldsymbol{\theta})\boldsymbol{Q}_{\beta
}(\boldsymbol{\theta})\right]  ^{-1}\boldsymbol{Q}_{\beta}^{T}%
(\boldsymbol{\theta})\boldsymbol{U}_{\beta,G}(\boldsymbol{\tilde{T}}_{\beta
,G}). \label{EQ:CompTest _Func}%
\end{equation}
At the null hypothesis in (\ref{1}), we also have $G=F_{\boldsymbol{\theta}}$
for some $\boldsymbol{\theta}\in\Theta_{0}$ and then by Fisher consistency of
the RMDPDE $\boldsymbol{\tilde{T}}_{\beta,F_{\boldsymbol{\theta}}%
}=\boldsymbol{\theta}$ and $\boldsymbol{U}_{\beta,F_{\boldsymbol{\theta}}%
}(\boldsymbol{\tilde{T}}_{\beta,F_{\boldsymbol{\theta}}})=\boldsymbol{0}_{p}$
implying $\tilde{R}_{\beta,F_{\boldsymbol{\theta}}}(\boldsymbol{\tilde{T}%
}_{\beta,F_{\boldsymbol{\theta}}})=0$.

Now, in order to derive the influence function for these Rao-type test
functionals $R_{\beta,G}(\boldsymbol{\theta}_{0})$ and $\tilde{R}_{\beta
,G}(\boldsymbol{\tilde{T}}_{\beta,G})$, we consider the contaminated
distribution $G_{\epsilon,y}=(1-\epsilon)G+\epsilon\Lambda_{y}$ having density
$g_{\epsilon}$, where $\epsilon$ is the contamination proportion and
$\Lambda_{y}$ is the degenerate distribution function at the contamination
point $y$. Then, the classical (first order) influence function of
$R_{\beta,G}(\boldsymbol{\theta}_{0})$ and $\tilde{R}_{\beta,G}%
(\boldsymbol{\tilde{T}}_{\beta,G})$ at the true distribution $G$ are,
respectively, given by
\[
\mathcal{IF}(y,R_{\beta,G}(\boldsymbol{\theta}_{0}))=\left.  \frac{\partial
}{\partial\epsilon}R_{\beta,G_{\epsilon,y}}(\boldsymbol{\theta}_{0}%
)\right\vert _{\epsilon=0},~~~~\mathcal{IF}(y,\tilde{R}_{\beta,G}%
(\boldsymbol{\tilde{T}}_{\beta,G}))=\left.  \frac{\partial}{\partial\epsilon
}\tilde{R}_{\beta,G_{\epsilon,y}}(\boldsymbol{\tilde{T}}_{\beta,G_{\epsilon
,y}})\right\vert _{\epsilon=0}.
\]
Now, first consider the case of $R_{0}^{\beta}$ to check that
\begin{align}
\mathcal{IF}(y,R_{\beta,G}(\boldsymbol{\theta}_{0}))  &  =2\boldsymbol{U}%
_{\beta,G}^{T}\left(  \boldsymbol{\theta}_{0}\right)  \boldsymbol{K}_{\beta
}^{-1}\left(  \boldsymbol{\theta}_{0}\right)  \left.  \frac{\partial}%
{\partial\epsilon}\boldsymbol{U}_{\beta,G_{\epsilon,y}}\left(
\boldsymbol{\theta}_{0}\right)  \right\vert _{\epsilon=0}\nonumber\\
&  =2\boldsymbol{U}_{\beta,G}^{T}\left(  \boldsymbol{\theta}_{0}\right)
\boldsymbol{K}_{\beta}^{-1}\left(  \boldsymbol{\theta}_{0}\right)
\boldsymbol{u}_{\beta}\left(  y,\boldsymbol{\theta}_{0}\right)  .\nonumber
\end{align}
When evaluated at the null hypothesis, as is the usual practice, we get
$\mathcal{IF}(y,R_{F_{\boldsymbol{\theta}_{0}}}^{\beta}(\boldsymbol{\theta
}_{0}))=0$, since $\boldsymbol{U}_{F_{\boldsymbol{\theta}_{0}}}^{\beta
}(\boldsymbol{\theta}_{0})=\boldsymbol{0}_{p}$. Similarly, one can also check
in the case of composite hypothesis that $\mathcal{IF}(y,\tilde{R}%
_{F_{\boldsymbol{\theta}_{0}}}^{\beta}(\boldsymbol{\tilde{T}}_{\beta
,F_{\boldsymbol{\theta}_{0}}}))=0$. Thus, the first order influence function
is inadequate to indicate the robustness properties of the proposed Rao-type
tests. This is consistent with other quadratic tests in the literature and we
need to consider the second order influence function.

The second order influence function of our Rao-type test functionals
$R_{\beta,G}(\boldsymbol{\theta}_{0})$ and $\tilde{R}_{\beta,G}%
(\boldsymbol{\tilde{T}}_{\beta,G})$ are similarly defined as
\[
\mathcal{IF}_{2}(y,R_{\beta,G}(\boldsymbol{\theta}_{0}))=\left.
\frac{\partial^{2}}{\partial\epsilon^{2}}R_{\beta,G_{\epsilon,y}%
}(\boldsymbol{\theta}_{0})\right\vert _{\epsilon=0},~~~~\mathcal{IF}%
_{2}(y,\tilde{R}_{\beta,G}(\boldsymbol{\tilde{T}}_{\beta,G}))=\left.
\frac{\partial^{2}}{\partial\epsilon^{2}}\tilde{R}_{\beta,G_{\epsilon,y}%
}(\boldsymbol{\tilde{T}}_{\beta,G_{\epsilon,y}})\right\vert _{\epsilon=0}.
\]
The following theorem presents their explicit forms at the corresponding null
hypotheses; the proof is straightforward from the previous lines and are hence
omitted.

\begin{theorem}
\label{THM:IF} Under the assumptions of Sections 2 and 3, the following
results hold.

\begin{enumerate}
\item For the simple hypothesis testing problem in (\ref{2.1}), we have
\[
\mathcal{IF}_{2}(y,R_{\beta,G}(\boldsymbol{\theta}_{0}))=2\boldsymbol{u}%
_{\beta}^{T}\left(  y,\boldsymbol{\theta}_{0}\right)  \boldsymbol{K}_{\beta
}^{-1}\left(  \boldsymbol{\theta}_{0}\right)  \boldsymbol{u}_{\beta}\left(
y,\boldsymbol{\theta}_{0}\right)  .
\]

\item For any $\boldsymbol{\theta}\in\Theta_{0}$ we have, in the composite
hypothesis testing problem,
\begin{align*}
&  \mathcal{IF}_{2}(y,\tilde{R}_{\beta,G}(\boldsymbol{\tilde{T}}_{\beta,G}))\\
&  =2\mathcal{IF}^{T}(y,\boldsymbol{\tilde{T}}_{\beta,G})\boldsymbol{Q}%
_{\beta}(\boldsymbol{\theta})\left[  \boldsymbol{Q}_{\beta}^{T}%
(\boldsymbol{\theta})\boldsymbol{K}_{\beta}(\boldsymbol{\theta})\boldsymbol{Q}%
_{\beta}(\boldsymbol{\theta})\right]  ^{-1}\boldsymbol{Q}_{\beta}%
^{T}(\boldsymbol{\theta})\mathcal{IF}(y,\boldsymbol{\tilde{T}}_{\beta,G}),
\end{align*}
where $\mathcal{IF}(y,\boldsymbol{\tilde{T}}_{\beta,G})$ is (\ref{IFEst}). In
particular at the null hypothesis $G=F_{\boldsymbol{\theta}}$, we have
$\mathcal{IF}(y,\boldsymbol{\tilde{T}}_{\beta,F_{\boldsymbol{\theta}}%
})=\mathcal{IF}(y,\boldsymbol{\theta})=\boldsymbol{J}_{\beta}^{-1}%
({{\boldsymbol{\theta}}})\boldsymbol{u}_{\beta}\left(  y,\boldsymbol{\theta
}\right)  $.
\end{enumerate}
\end{theorem}

They are bounded, implying robustness whenever the quantity $\boldsymbol{u}%
_{\beta}\left(  y,\boldsymbol{\theta}\right)  $ is bounded at the
contamination point $y$; this holds for all $\beta>0$ in most statistical
models. Thus, our proposed Rao-type test statistics are robust at all
$\beta>0$. But, at $\beta=0$ the corresponding influence function is unbounded
indicating the well-known non-robust nature of the classical Rao test.

\subsection{Power and Level influence function}

We will now consider the asymptotic level and asymptotic contiguous power of
the proposed Rao-type tests and their robustness with respect to contiguous
contamination. Suppose $F_{\boldsymbol{\theta}_{0}}$ be the true data
generating density under the null hypothesis, either given in (\ref{2.2}) or
in (\ref{1}) where $\boldsymbol{\theta}_{0}$ is a fix value of the parameter
under the null hypothesis, and consider the contiguous alternative hypotheses
of the form in (\ref{2.4}), i.e., $\boldsymbol{\theta}_{n,\boldsymbol{d}%
}=\boldsymbol{\theta}_{0}+\tfrac{1}{\sqrt{n}}\boldsymbol{d}$. We then consider
the contiguous contaminated sequences of distributions
\[
F_{n,\epsilon,y}^{L}=(1-\tfrac{\epsilon}{\sqrt{n}}%
)F_{\boldsymbol{{\boldsymbol{\theta}}}_{0}}+\tfrac{\epsilon}{\sqrt{n}}%
\Lambda_{y}\quad\text{and}\quad F_{n,\epsilon,y,\boldsymbol{d}}^{P}%
=(1-\tfrac{\epsilon}{\sqrt{n}})F_{\boldsymbol{\theta}_{n,\boldsymbol{d}}%
}+\tfrac{\epsilon}{\sqrt{n}}\Lambda_{y},
\]
and study the influence of contamination on the asymptotic level and power of
the proposed Rao-type tests based on $R_{\beta,n}\left(  \boldsymbol{\theta
}_{0}\right)  $, respectively under these contaminated distributions, given
by
\[
\alpha_{\beta,\epsilon,y}(\boldsymbol{\theta}_{0})=\lim\limits_{n\rightarrow
\infty}P_{F_{n,\epsilon,y}^{L}}(R_{\beta,n}(\boldsymbol{\theta}_{0}%
)>\chi_{p,\alpha}^{2})\quad\text{and}\quad{\pi}_{\beta,\epsilon
,y,\boldsymbol{d}}(\boldsymbol{\theta}_{0})=\lim\limits_{n\rightarrow\infty
}P_{F_{n,\epsilon,y,\boldsymbol{d}}^{P}}(R_{\beta,n}(\boldsymbol{\theta}%
_{0})>\chi_{p,\alpha}^{2}).
\]
Then, the level influence function (LIF) and the power influence function
(PIF) are defined as
\[
\mathcal{LIF}(y,R_{\beta,F_{{\boldsymbol{\theta}}_{0}}}(\boldsymbol{\theta
}_{0}))=\left.  \dfrac{\partial}{\partial\epsilon}\alpha_{\beta,\epsilon
,y}(\boldsymbol{\theta}_{0})\right\vert _{\epsilon=0}\quad\text{and}%
\quad\mathcal{PIF}(y,\boldsymbol{d},R_{\beta,F_{{\boldsymbol{\theta}}_{0}}%
}(\boldsymbol{\theta}_{0}))=\left.  \dfrac{\partial}{\partial\epsilon}{\pi
}_{\beta,\epsilon,y,\boldsymbol{d}}(\boldsymbol{\theta}_{0})\right\vert
_{\epsilon=0}.
\]
We will now explicitly derive the form of these LIF and PIF for our Rao-type
tests for testing both the simple and composite null hypotheses, and study
their boundedness over the contamination point $y$.

Let us start with simple null hypotheses given in (\ref{2.2}), and derive the
corresponding asymptotic power ${\pi}_{\beta,\epsilon,y,\boldsymbol{d}%
}(\boldsymbol{\theta}_{0})$ under the contiguous contaminated distribution
$F_{n,\epsilon,y,\boldsymbol{d}}^{P}$.

\begin{theorem}
\label{THM:7asymp_power_one} Consider testing the simple null hypothesis
(\ref{2.2}) by the Rao-type test statistics $R_{\beta,n}(\boldsymbol{\theta
}_{0})$ at $\alpha$-level of significance. Then the following results hold.

\begin{enumerate}
\item[i)] The asymptotic distribution of $R_{\beta,n}({\boldsymbol{\theta}%
}_{0})$ under $F_{n,\epsilon,y,\boldsymbol{d}}^{P}$ is a non-central
chi-square distribution with $p$ degrees of freedom and non-centrality
parameter given by
\[
\delta_{\beta,\epsilon,y}(\boldsymbol{\theta}_{0},\boldsymbol{d}%
)=\boldsymbol{\delta}_{\beta,\epsilon,y}^{T}(\boldsymbol{\theta}%
_{0},\boldsymbol{d})\boldsymbol{J}_{\beta}\left(  \boldsymbol{\theta}%
_{0}\right)  \boldsymbol{K}_{\beta}^{-1}\left(  \boldsymbol{\theta}%
_{0}\right)  \boldsymbol{J}_{\beta}\left(  \boldsymbol{\theta}_{0}\right)
\boldsymbol{\delta}_{\beta,\epsilon,y}(\boldsymbol{\theta}_{0},\boldsymbol{d}%
),
\]
where
\[
\boldsymbol{\delta}_{\beta,\epsilon,y}(\boldsymbol{\theta}_{0},\boldsymbol{d}%
)=\boldsymbol{d}+\epsilon\mathcal{IF}(y,\boldsymbol{T}_{\beta
,F_{{\boldsymbol{\theta}}_{0}}})
\]
and $\mathcal{IF}(y;\boldsymbol{T}_{\beta},F_{{\boldsymbol{\theta}}_{0}})$ is
the influence function of the MDPDE as given by (Basu et al., 1998, 2011)
\[
\mathcal{IF}(y,\boldsymbol{T}_{\beta,F_{{\boldsymbol{\theta}}_{0}}%
})=\boldsymbol{J}_{\beta}^{-1}(\boldsymbol{\theta}_{0})\boldsymbol{u}_{\beta
}\left(  y,\boldsymbol{\theta}_{0}\right)  .
\]

\item[ii)] The corresponding asymptotic power under contiguous contaminated
sequence of alternative distributions $F_{n,\epsilon,y}^{P}$ is given by
\begin{equation}
{\pi}_{\beta,\epsilon,y,\boldsymbol{d}}(\boldsymbol{\theta}_{0})=\sum
\limits_{k=0}^{\infty}c_{k}\left(  \delta_{\beta,\epsilon,y}%
(\boldsymbol{\theta}_{0},\boldsymbol{d})\right)  P\left(  \chi_{p+2k}^{2}%
>\chi_{p,\alpha}^{2}\right)  , \label{EQ:Cont_power_one}%
\end{equation}
where $c_{k}\left(  s\right)  =\frac{s^{k}}{k!2^{k}}e^{-\frac{s}{2}}$.

\end{enumerate}
\end{theorem}

\begin{proof}
Denote $\boldsymbol{\theta}_{n,\epsilon,\boldsymbol{d}}^{\ast}=\boldsymbol{T}%
_{\beta,F_{n,\epsilon,y}^{P}}$. Then, considering $\boldsymbol{\theta
}_{n,\epsilon,\boldsymbol{d}}^{\ast}$ as a function of $\tfrac{\epsilon}%
{\sqrt{n}}$ and using Taylor series expansion of $\boldsymbol{\theta
}_{n,\epsilon,\boldsymbol{d}}^{\ast}$ at $\tfrac{\epsilon}{\sqrt{n}}=0$
($\epsilon=0$), we get
\[
\boldsymbol{\theta}_{n,\epsilon,\boldsymbol{d}}^{\ast}=\boldsymbol{\theta
}_{n,\boldsymbol{d}}+\tfrac{\epsilon}{\sqrt{n}}\mathcal{IF}(y,\boldsymbol{T}%
_{\beta,F_{{\boldsymbol{\theta}}_{0}}})+o_{P}(\tfrac{1}{\sqrt{n}%
}\boldsymbol{1}_{p}).
\]
Hence,
writing $\boldsymbol{\theta}_{n,\boldsymbol{d}}$ in terms of
$\boldsymbol{\theta}_{0}$, we get
\begin{equation}
\sqrt{n}(\boldsymbol{\theta}_{n,\epsilon,\boldsymbol{d}}^{\ast}%
-\boldsymbol{\theta}_{0})=\boldsymbol{d}+\epsilon\mathcal{IF}(y,\boldsymbol{T}%
_{\beta,F_{{\boldsymbol{\theta}}_{0}}})+o_{P}(\boldsymbol{1}_{p}%
)=\boldsymbol{\delta}_{\beta,\epsilon,y}(\boldsymbol{\theta}_{0}%
)+o_{P}(\boldsymbol{1}_{p}). \label{DT}%
\end{equation}
Next, considering Taylor series expansion of $\boldsymbol{U}_{\beta
,n}(\boldsymbol{\theta}_{n,\epsilon,\boldsymbol{d}}^{\ast})$ at the point
$\boldsymbol{\theta}_{0}$, we get%
\[
\boldsymbol{U}_{\beta,n}(\boldsymbol{\theta}_{n,\epsilon,\boldsymbol{d}}%
^{\ast})=\boldsymbol{U}_{\beta,n}(\boldsymbol{\theta}_{0})+\left.
\frac{\partial}{\partial\boldsymbol{\theta}}\boldsymbol{U}_{\beta,n}%
^{T}(\boldsymbol{\theta})\right\vert _{\boldsymbol{\theta}=\boldsymbol{\theta
}_{n,\epsilon,\boldsymbol{d}}^{\ast\ast}}(\boldsymbol{\theta}_{n,\epsilon
,\boldsymbol{d}}^{\ast}-\boldsymbol{\theta}_{0}),
\]
where $\boldsymbol{\theta}_{n,\epsilon,\boldsymbol{d}}^{\ast\ast}$ belongs to
the line segment joining $\boldsymbol{\theta}_{0}$ and $\boldsymbol{\theta
}_{n,\epsilon,\boldsymbol{d}}^{\ast}$. Then, we proceed as in the proof of
Theorem \ref{THM:cont_1} and use (\ref{DT}) to conclude the first part of the
Theorem. Here, in order to use limit theorems we need to note that,
$\boldsymbol{\theta}_{n,\epsilon,\boldsymbol{d}}^{\ast}$ is contiguous to
$\boldsymbol{\theta}_{0}$ under $F_{n,\epsilon,y,\boldsymbol{d}}^{P}$ by
(\ref{DT}) and hence we can apply Le Cam's Third Lemma (see, eg., van der
Vaart, 1990, p.~90), and we need to use the continuity of the matrices
$\boldsymbol{J}_{\beta}(\boldsymbol{\theta})$ and $\boldsymbol{K}_{\beta
}(\boldsymbol{\theta})$ at $\boldsymbol{\theta}=\boldsymbol{\theta}_{0}%
$.\newline Next, to prove Part (ii) of the theorem, we directly use the
infinite series expansion of non-central chi-square distribution functions
(Kotz et al., 1967b) in terms of those of independent central chi-square
variables $\chi_{p+2k}^{2}$, $k=0,1,2,\ldots$, as follows
\begin{align}
{\pi}_{\beta,\epsilon,y}(\boldsymbol{\theta}_{0},\boldsymbol{d})  &
=\lim\limits_{n\rightarrow\infty}P_{F_{n,\epsilon,y,\boldsymbol{d}}^{P}%
}(R_{\beta,n}(\boldsymbol{\theta}_{0})>\chi_{p,\alpha}^{2})\nonumber\\
&  =P(\chi_{p}^{2}\left(  \delta_{\beta,\epsilon,y}(\boldsymbol{\theta}%
_{0},\boldsymbol{d})\right)  >\chi_{p,\alpha}^{2})\nonumber\\
&  =\sum\limits_{k=0}^{\infty}c_{k}\left(  \delta_{\beta,\epsilon
,y}(\boldsymbol{\theta}_{0},\boldsymbol{d})\right)  P\left(  \chi_{p+2k}%
^{2}>\chi_{p,\alpha}^{2}\right)  .\nonumber
\end{align}

\end{proof}

At the special case $\epsilon=0$, the first part of the above theorem coincide
with Theorem \ref{THM:cont_1} (noting $\delta_{\epsilon=0,y}^{\beta
}(\boldsymbol{\theta}_{0},\boldsymbol{d})=\delta_{\beta}(\boldsymbol{\theta
}_{0},\boldsymbol{d})$ given in (\ref{cp})) and the second part then gives an
infinite series expression for the asymptotic contiguous power of our Rao-type
tests as
\begin{equation}
{\pi}_{\beta}(\boldsymbol{\theta}_{0},\boldsymbol{d})={\pi}_{\beta
,\epsilon=0,y}(\boldsymbol{\theta}_{0},\boldsymbol{d})=\sum\limits_{k=0}%
^{\infty}c_{k}\left(  \delta_{\beta}(\boldsymbol{\theta}_{0},\boldsymbol{d}%
)\right)  P\left(  \chi_{p+2k}^{2}>\chi_{p,\alpha}^{2}\right)  .\nonumber
\end{equation}
However, substituting $\boldsymbol{d}=\boldsymbol{0}_{p}$ in Theorem
\ref{THM:7asymp_power_one}, Expression (\ref{EQ:Cont_power_one}) yields the
asymptotic level under the contaminated distribution $F_{n,\epsilon,y}^{L}$ as
given by
\begin{align}
\alpha_{\beta,\epsilon,y}(\boldsymbol{\theta}_{0})  &  ={\pi}_{\beta
,\epsilon,y}(\boldsymbol{\theta}_{0},\boldsymbol{d}=\boldsymbol{0}%
_{p})\nonumber\\
&  =\sum\limits_{k=0}^{\infty}c_{k}\left(  \epsilon^{2}\boldsymbol{u}_{\beta
}^{T}\left(  y,\boldsymbol{\theta}_{0}\right)  \boldsymbol{K}_{\beta}%
^{-1}\left(  \boldsymbol{\theta}_{0}\right)  \boldsymbol{u}_{\beta}\left(
y,\boldsymbol{\theta}_{0}\right)  \right)  P\left(  \chi_{p+2k}^{2}%
>\chi_{p,\alpha}^{2}\right)  .\nonumber
\end{align}

Finally, the PIF can be obtained by an appropriate differentiation of the
asymptotic power ${\pi}_{\beta,\epsilon,y}(\boldsymbol{\theta}_{0}%
,\boldsymbol{d})$ from the expression (\ref{EQ:Cont_power_one}) of Theorem
\ref{THM:7asymp_power_one}. The LIF can also be obtained similarly from ${\pi
}_{\beta,\epsilon,y}(\boldsymbol{\theta}_{0},\boldsymbol{d})$ or by
substituting $\boldsymbol{d}=\boldsymbol{0}_{p}$ in the formula of PIF. The
final expression is given in the following theorem but the proof is omitted
for brevity; see Ghosh et al.~(2016) for similar calculations.

\begin{theorem}
\label{THM:PIF_one} Under the assumptions of Theorem
\ref{THM:7asymp_power_one}, the power and level influence function of the
Rao-type tests for testing the simple null hypothesis in (\ref{2.2}) is given
by
\[
\mathcal{PIF}(y,\boldsymbol{d},R_{\beta,F_{{\boldsymbol{\theta}}_{0}}%
}({\boldsymbol{\theta}}_{0}))=C_{p}\left(  \delta_{\beta}(\boldsymbol{\theta
}_{0},\boldsymbol{d})\right)  \boldsymbol{d}^{T}\boldsymbol{J}_{\beta}\left(
\boldsymbol{\theta}_{0}\right)  \boldsymbol{K}_{\beta}^{-1}\left(
\boldsymbol{\theta}_{0}\right)  \boldsymbol{u}_{\beta}\left(
y,\boldsymbol{\theta}_{0}\right)  ,
\]
with $\delta_{\beta}(\boldsymbol{\theta}_{0},\boldsymbol{d})$ given by
(\ref{cp}),
\[
C_{p}(s)=e^{-\frac{s}{2}}\sum\limits_{k=0}^{\infty}\frac{s^{k-1}}{k!2^{k}%
}(2k-s)P\left(  \chi_{p+2k}^{2}>\chi_{p,\tau}^{2}\right)  ,
\]
and
\[
\mathcal{LIF}(y,R_{\beta,F_{{\boldsymbol{\theta}}_{0}}}({\boldsymbol{\theta}%
}_{0}))=0.
\]

\end{theorem}

\bigskip

Interestingly, we can see from the above theorem that the PIF of our proposed
Rao-type test will be bounded in the contamination point $y$ if and only if
the expression $\boldsymbol{u}_{\beta}\left(  y,\boldsymbol{\theta}%
_{0}\right)  $ is so, which is known to be true for most models at $\beta>0$.
Hence, the asymptotic power of the proposed Rao-type tests at any $\beta>0$
would be stable under infinitesimal contaminations and that of the classical
Rao test (at $\beta=0$) would be non-robust having an unbounded PIF. Note the
similarity with the boundedness of the (second order) IF of the proposed
Rao-type test statistics discussed in the previous subsection. Further, the
asymptotic level of our proposed Rao-type tests would also be extremely robust
against infinitesimal contiguous contaminations having bounded (zero) LIF.

The case of composite hypothesis in (\ref{1}) can be studied similarly. The
asymptotic contiguous power and level under contiguous contaminations can be
derived in a similar fashion and hence the corresponding PIF and LIF at
$\boldsymbol{\theta}\in\boldsymbol{\Theta}_{0}$ based on%
\begin{align*}
&  \tilde{\alpha}_{\beta,\epsilon,y}(\boldsymbol{\theta})=\lim
\limits_{n\rightarrow\infty}P_{F_{n,\epsilon,y}^{L}}(\widetilde{R}_{\beta
,n}(\boldsymbol{\tilde{\theta}})>\chi_{p,\alpha}^{2})\quad\text{and}%
\quad{\tilde{\pi}}_{\beta,\epsilon,y,\boldsymbol{d}}(\boldsymbol{\theta}%
)=\lim\limits_{n\rightarrow\infty}P_{F_{n,\epsilon,y,\boldsymbol{d}}^{P}%
}(\widetilde{R}_{\beta,n}(\boldsymbol{\tilde{\theta}})>\chi_{p,\alpha}^{2}),\\
&  \mathcal{LIF}(y,\tilde{R}_{\beta,F_{\boldsymbol{{\boldsymbol{\theta}}}}%
}(\boldsymbol{\tilde{T}}_{\beta,F_{\boldsymbol{{\boldsymbol{\theta}}}}%
}))=\left.  \dfrac{\partial}{\partial\epsilon}\tilde{\alpha}_{\beta
,\epsilon,y}(\boldsymbol{\theta})\right\vert _{\epsilon=0}\quad\text{and}%
\quad\mathcal{PIF}(y,\boldsymbol{d},\tilde{R}_{\beta
,F_{\boldsymbol{{\boldsymbol{\theta}}}}}(\boldsymbol{\tilde{T}}_{\beta
,F_{\boldsymbol{{\boldsymbol{\theta}}}}}))=\left.  \dfrac{\partial}%
{\partial\epsilon}{\tilde{\pi}}_{\beta,\epsilon,y,\boldsymbol{d}%
}(\boldsymbol{\theta})\right\vert _{\epsilon=0}.
\end{align*}
The main results are presented in the following two theorems, but their proofs
are omitted for brevity. The implications are again the same.

\begin{theorem}
\label{THM:7asymp_power_composite} Consider testing the composite null
hypothesis in (\ref{1}) by the Rao-type test statistics $\widetilde{R}%
_{\beta,n}(\widetilde{\boldsymbol{\theta}}_{\beta})$ at $\alpha$ level of
significance and $\boldsymbol{\theta}\in\boldsymbol{\Theta}_{0}$. Then the
following results hold.

\begin{enumerate}
\item[i)] The asymptotic distribution of $\widetilde{R}_{\beta,n}%
(\widetilde{\boldsymbol{\theta}}_{\beta})$ under $F_{n,\epsilon
,y,\boldsymbol{d}}^{P}$ is a non-central chi-square distribution with $r$
degrees of freedom and non-centrality parameter given by
\[
\tilde{\delta}_{\beta,\epsilon,y}(\boldsymbol{\theta},\boldsymbol{d}%
)=\boldsymbol{\delta}_{\beta,\epsilon,y}^{T}(\boldsymbol{\theta}%
,\boldsymbol{d})\boldsymbol{J}_{\beta}\left(  \boldsymbol{\theta}\right)
\boldsymbol{Q}_{\beta}(\boldsymbol{\theta})\left[  \boldsymbol{Q}_{\beta}%
^{T}(\boldsymbol{\theta})\boldsymbol{K}_{\beta}(\boldsymbol{\theta
})\boldsymbol{Q}_{\beta}(\boldsymbol{\theta})\right]  ^{-1}\boldsymbol{Q}%
_{\beta}^{T}(\boldsymbol{\theta})\boldsymbol{J}_{\beta}\left(
\boldsymbol{\theta}\right)  \boldsymbol{\delta}_{\beta,\epsilon,y}%
(\boldsymbol{\theta},\boldsymbol{d}),
\]
where $\boldsymbol{\delta}_{\beta,\epsilon,y}(\boldsymbol{\theta
},\boldsymbol{d})$ is as in Theorem \ref{THM:7asymp_power_one}.

\item[ii)] The corresponding asymptotic power under contiguous contaminated
sequence of alternative distributions $F_{n,\epsilon,y,\boldsymbol{d}}^{P}$ is
given by
\[
{\pi}_{\beta,\epsilon,y,\boldsymbol{d}}(\boldsymbol{\theta})=\sum
\limits_{k=0}^{\infty}c_{k}\left(  \tilde{\delta}_{\beta,\epsilon
,y}(\boldsymbol{\theta},\boldsymbol{d})\right)  P\left(  \chi_{r+2k}^{2}%
>\chi_{r,\alpha}^{2}\right)  ,
\]
where $c_{k}\left(  s\right)  $ is as defined in Theorem
\ref{THM:7asymp_power_one}.
\end{enumerate}
\end{theorem}

\bigskip

\begin{theorem}
\label{THM:PIF_composite} Under the assumptions of Theorem
\ref{THM:7asymp_power_composite}, the power and level influence function of
the Rao-type tests for testing the simple null hypothesis in (\ref{1}) is
given by
\begin{align*}
\mathcal{PIF}(y,\boldsymbol{d},\tilde{R}_{\beta
,F_{\boldsymbol{{\boldsymbol{\theta}}}}}(\boldsymbol{\tilde{T}}_{\beta
,F_{\boldsymbol{{\boldsymbol{\theta}}}}}))  &  =C_{p}\left(  \tilde{\delta
}_{\beta}(\boldsymbol{\theta},\boldsymbol{d})\right)  \boldsymbol{d}%
^{T}\boldsymbol{J}_{\beta}\left(  \boldsymbol{\theta}\right)  \boldsymbol{Q}%
_{\beta}(\boldsymbol{\theta})\left[  \boldsymbol{Q}_{\beta}^{T}%
(\boldsymbol{\theta})\boldsymbol{K}_{\beta}(\boldsymbol{\theta})\boldsymbol{Q}%
_{\beta}(\boldsymbol{\theta})\right]  ^{-1}\\
&  \times\boldsymbol{Q}_{\beta}^{T}(\boldsymbol{\theta})\boldsymbol{u}_{\beta
}\left(  y,\boldsymbol{\theta}\right)  ,
\end{align*}
where $C_{p}(s)$ is as defined in Theorem \ref{THM:PIF_one},%
\[
\tilde{\delta}_{\beta}(\boldsymbol{\theta},\boldsymbol{d})=\boldsymbol{d}%
^{T}\boldsymbol{J}_{\beta}\left(  \boldsymbol{\theta}\right)  \boldsymbol{Q}%
_{\beta}(\boldsymbol{\theta})\left[  \boldsymbol{Q}_{\beta}^{T}%
(\boldsymbol{\theta})\boldsymbol{K}_{\beta}(\boldsymbol{\theta})\boldsymbol{Q}%
_{\beta}(\boldsymbol{\theta})\right]  ^{-1}\boldsymbol{Q}_{\beta}%
^{T}(\boldsymbol{\theta})\boldsymbol{J}_{\beta}\left(  \boldsymbol{\theta
}\right)  \boldsymbol{d}%
\]
and
\[
\mathcal{LIF}(y,\tilde{R}_{\beta,F_{\boldsymbol{{\boldsymbol{\theta}}}}%
}(\boldsymbol{\tilde{T}}_{\beta,F_{\boldsymbol{{\boldsymbol{\theta}}}}}))=0.
\]

\end{theorem}


\section{Examples}

We consider the data on telephone line faults analyzed, among others, in Basu
et al. (2013). The data are presented in Table \ref{tel} and consist of the
ordered differences between the inverse test rates and the inverse control
rates in matched pairs of areas. The first observation of this dataset is a
huge outlier with respect to the normal model. In the following three
subsections three different hypothesis testing problems are considered to be
analyzed with this data set.%

\begin{table}[htbp]  \tabcolsep2.8pt  \centering
\begin{tabular}
[c]{lcccccccccccccc}\hline
pairs & $1$ & $2$ & $3$ & $4$ & $5$ & $6$ & $7$ & $8$ & $9$ & $10$ & $11$ &
$12$ & $13$ & $14$\\
differences & $-988$ & $-135$ & $-78$ & $3$ & $59$ & $83$ & $93$ & $110$ &
$189$ & $197$ & $204$ & $229$ & $289$ & $310$\\\hline
\end{tabular}
\caption{Telephone-line faults data\label{tel}}
\end{table}%


\subsection{Example 1 (One-dimensional simple null hypothesis)}

\label{Example1}

Suppose we have an i.i.d. sample $X_{1},\ldots,X_{n}$ from a normal population
with variance $\sigma_{0}^{2}$ and unknown mean $\mu$. We want to test the
hypothesis
\[
H_{0}:\mu=\mu_{0}\text{ against }H_{1}:\mu\neq\mu_{0}%
\]
using the proposed Rao-type test statistics. In this case the full parameter
space is given by $\Theta=\{\mu\in%
\mathbb{R}
\}$, while that under the null hypothesis is given by a unique point
$\Theta_{0}=\{\mu\in%
\mathbb{R}
:\mu=\mu_{0}\}$.
Direct calculations show that $u_{\beta}\left(  X,\mu\right)  $, as defined in
Equation (\ref{EQ:score_function}) is given by
\[
u_{\beta}\left(  x,\mu\right)  =\frac{x-\mu}{\sigma_{0}^{2}}\frac{1}{\left(
\sqrt{2\pi}\sigma_{0}\right)  ^{\beta}}\exp\left(  -\frac{\beta}{2}\left(
\frac{x-\mu}{\sigma_{0}}\right)  ^{2}\right)  .
\]
It is also easy to check
\[
u(x,\mu_{0})=u_{\beta=0}\left(  x,\mu_{0}\right)  =\frac{x-\mu_{0}}{\sigma
_{0}^{2}}~\text{and}~\xi\left(  \mu_{0}\right)  =\int_{-\infty}^{+\infty
}u(x,\mu_{0})f_{\mu_{0}}^{\beta+1}(x)dx=0.
\]
Also
\begin{align*}
K_{\beta}\left(  \mu_{0}\right)   &  =\int_{-\infty}^{+\infty}u^{2}\left(
x,\mu\right)  f_{\mu_{0}}^{2\beta+1}(x)dx-\xi^{2}\left(  \mu_{0}\right)  .\\
&  =\int_{-\infty}^{+\infty}\left(  \frac{x-\mu_{0}}{\sigma_{0}^{2}}\right)
^{2}\frac{1}{\left(  \sqrt{2\pi}\sigma_{0}\right)  ^{2\beta+1}}\exp\left(
-\frac{2\beta+1}{2}\left(  \frac{x-\mu_{0}}{\sigma_{0}}\right)  ^{2}\right)
dx-0\\
&  =\frac{1}{\left(  2\pi\right)  ^{\beta}\sigma_{0}^{2\beta+2}\left(
2\beta+1\right)  ^{\frac{3}{2}}}.
\end{align*}
Then
\begin{align}
R_{\beta,n}\left(  \mu_{0}\right)   &  =nU_{\beta}^{2}\left(  \mu\right)
\left/  K_{\beta}\left(  \mu_{0}\right)  \right.  =\frac{1}{nK_{\beta}\left(
\mu_{0}\right)  }\left(  \sum\limits_{i=1}^{n}u_{\beta}\left(  X_{i}%
,\mu\right)  \right)  ^{2}\nonumber\\
&  =\frac{\left(  2\pi\right)  ^{\beta}\sigma_{0}^{2\beta+2}}{n\left(
2\beta+1\right)  ^{-\frac{3}{2}}}\left[  \sum_{i=1}^{n}\frac{X_{i}-\mu_{0}%
}{\sigma_{0}^{\beta+2}}\frac{1}{\left(  2\pi\right)  ^{\frac{\beta}{2}}}%
\exp\left(  -\frac{\beta}{2}\left(  \frac{X_{i}-\mu_{0}}{\sigma_{0}}\right)
^{2}\right)  \right]  ^{2}\nonumber\\
&  =\frac{\left(  2\beta+1\right)  ^{\frac{3}{2}}}{n}\left(  \sum_{i=1}%
^{n}\frac{X_{i}-\mu_{0}}{\sigma_{0}}\exp\left(  -\frac{\beta}{2}\left(
\frac{X_{i}-\mu_{0}}{\sigma_{0}}\right)  ^{2}\right)  \right)  ^{2}.
\label{R1}%
\end{align}
For $\beta=0$, we recover the score equation for the MLE, and the test
statistic reduces to the classical Rao statistic
\[
R_{n}\left(  \mu_{0}\right)  =R_{\beta=0,n}\left(  \mu_{0}\right)  =\left(
\frac{\bar{X}_{n}-\mu_{0}}{\frac{\sigma_{0}}{\sqrt{n}}}\right)  ^{2}.
\]
where $\bar{X}_{n}$ is the sample mean.

In the top panel of Figure \ref{figEx1} the values of the Rao type test
statistics are plotted for the telephone-line faults data for $\mu_{0}=0$ and
$\sigma_{0}=175$. The threshold for the acceptance region of the null
hypothesis shows that for outlier deleted data, i.e. when the first
observation is removed, the null hypothesis is rejected for all values of
$\beta\in\lbrack0,1]$ at the nominal level $\alpha=0.05$; however, \ the full
data set, the null hypothesis cannot be rejected for most values of
$\beta<0.1$, including for the classical Rao test-statistic ($\beta=0$). The
difference in the final conclusion of the test in the full data and outlier
deleted data scenarios indicate the lack of robustness of the classical Rao
statistic, in that here a single extreme observation is able to control the
conclusion of the test in a sample of size $14$. On the other hand, values of
$\beta>0.1$ lead to a similar conclusion for the Rao type test statistic with
or without outlier; this indicates the outlier stability of the proposed Rao
type test statistics for moderately large values of $\beta$.%

\begin{figure}[htbp]  \tabcolsep2.8pt  \centering
\begin{tabular}
[c]{l}%
{\includegraphics[
trim=1.406426in 0.771690in 1.058220in 0.348137in,
height=4.0058in,
width=6.3763in
]%
{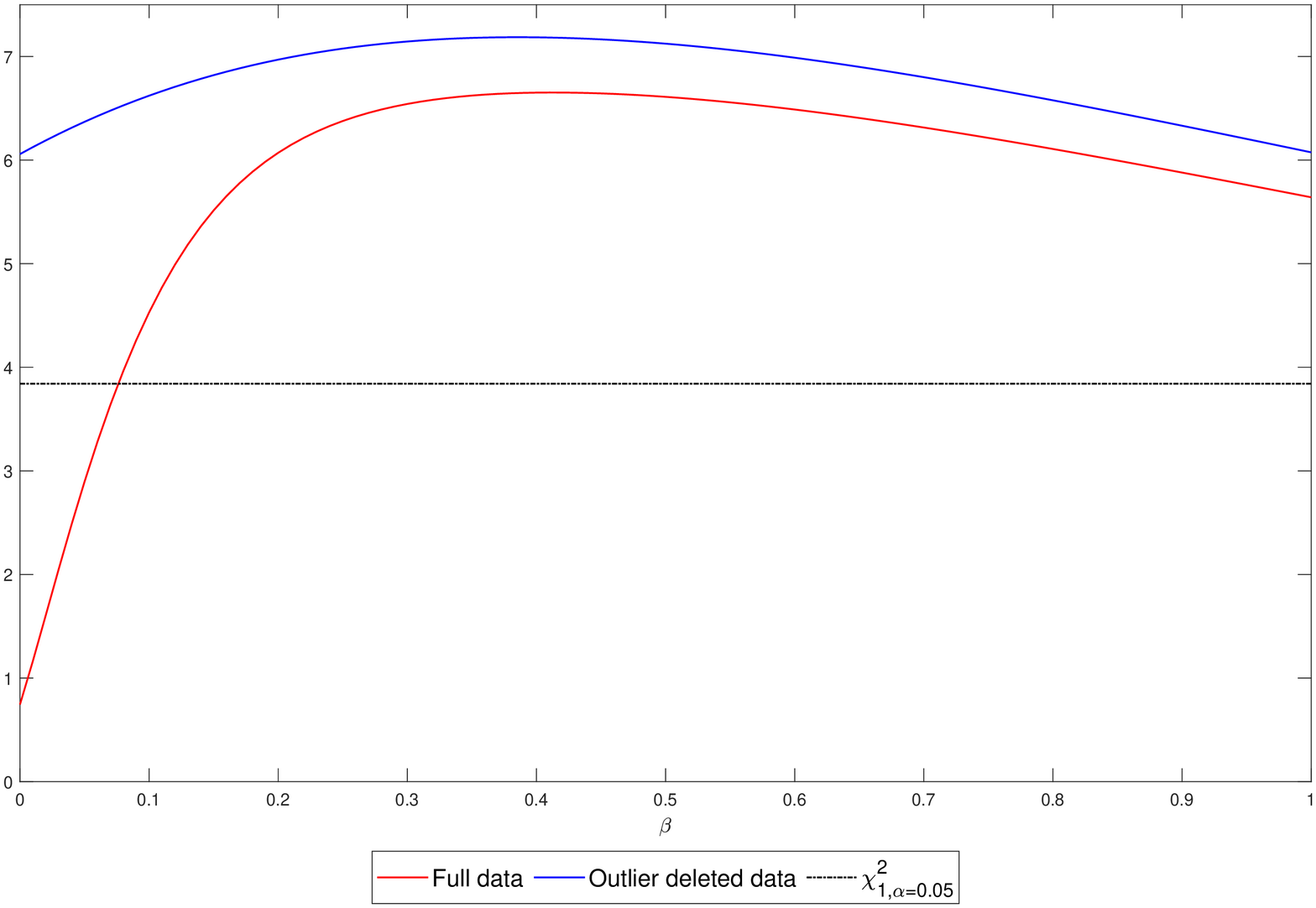}%
}
\\%
{\includegraphics[
trim=1.406426in 0.771690in 1.058220in 0.348137in,
height=4.0058in,
width=6.3763in
]%
{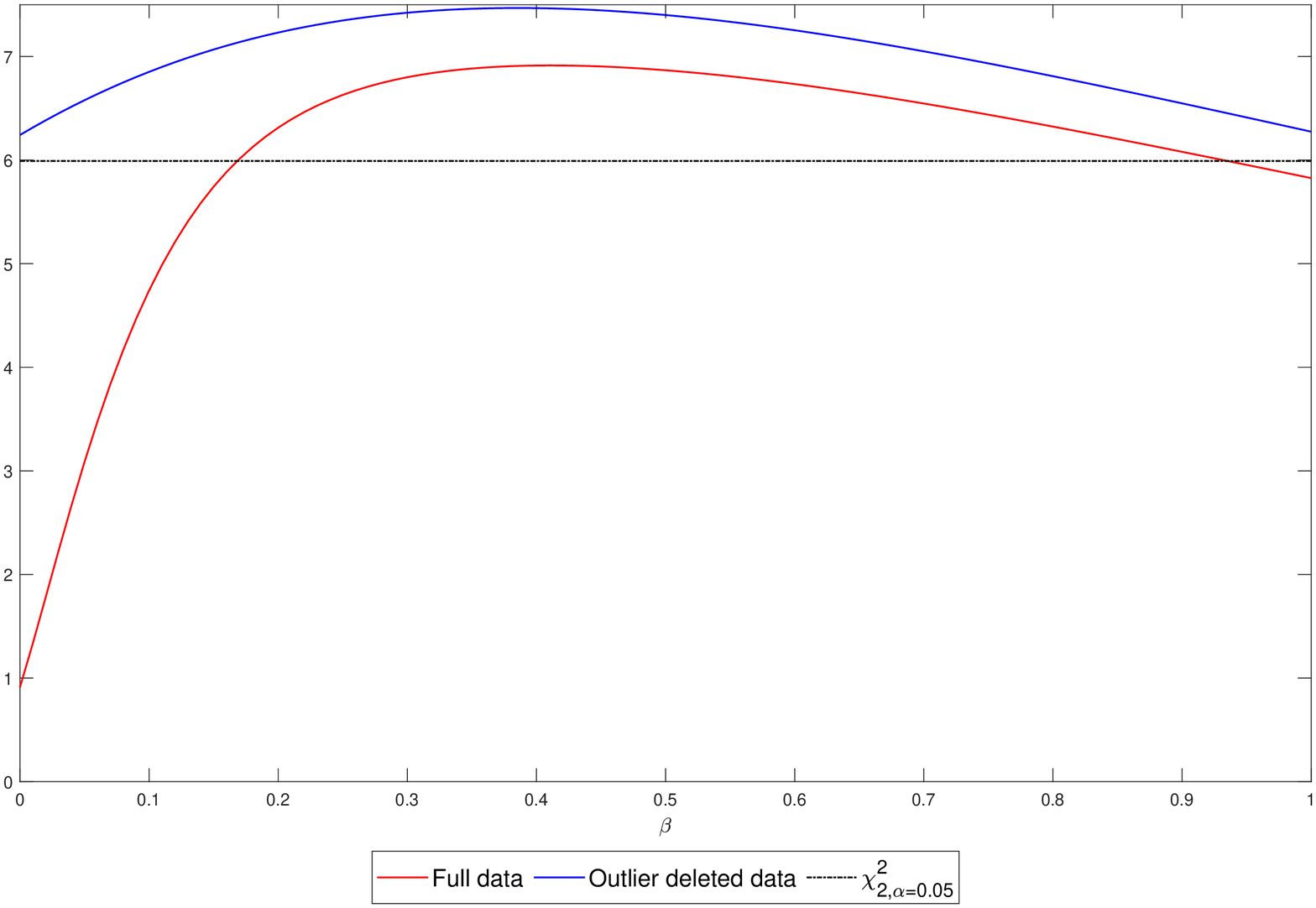}%
}
\end{tabular}
\caption{Telephone-line faults data: $R_{n,\beta}(\mu_{0})$ in the top and $R_{n,\beta}(\mu_{0},\sigma_{0})$ in the bottom, both in terms of tuning parameter $\beta$ on abscissa axis\label{figEx1}}
\end{figure}%

\subsection{Example 2 (Two-dimensional simple null hypothesis)}


Let $X$ be a normal population with both parameters $\mu$ and $\sigma$
unknown. We want to derive the Rao-type test statistics for testing%
\begin{equation}
H_{0}:\left(  \mu,\sigma\right)  =\left(  \mu_{0},\sigma_{0}\right)  \text{
against }H_{1}:\left(  \mu,\sigma\right)  \neq\left(  \mu_{0},\sigma
_{0}\right)  . \label{2.5}%
\end{equation}
In this case the full parameter space is given by $\Theta=\{\left(  \mu
,\sigma\right)  \in%
\mathbb{R}
\times%
\mathbb{R}
^{+}\}$, while that under the null hypothesis is given by a unique point
$\Theta_{0}=\{\left(  \mu,\sigma\right)  \in%
\mathbb{R}
\times%
\mathbb{R}
^{+}:(\mu,\sigma)=(\mu_{0},\sigma_{0})\}$. In order to evaluate
$\boldsymbol{u}_{\beta}\left(  x,\mu_{0},\sigma_{0}\right)  =\left(
u_{\beta,1}\left(  x,\mu_{0},\sigma_{0}\right)  ,u_{\beta,2}\left(  x,\mu
_{0},\sigma_{0}\right)  \right)  ^{T}$, note that from the previous example we
have
\begin{equation}
u_{\beta,1}\left(  x,\mu,\sigma\right)  =\frac{x-\mu}{\sigma^{\beta+2}}%
\frac{1}{\left(  2\pi\right)  ^{\frac{\beta}{2}}}\exp\left(  -\frac{\beta}%
{2}\left(  \frac{x-\mu}{\sigma}\right)  ^{2}\right)  . \label{EQ:phi_1}%
\end{equation}
Simple calculations yield
\begin{equation}
u_{\beta,2}\left(  x,\mu,\sigma\right)  =\frac{1}{\sigma^{\beta+1}\left(
2\pi\right)  ^{\beta/2}}\left[  \left(  \frac{(x-\mu)^{2}}{\sigma^{2}%
}-1\right)  \exp\left(  -\frac{\beta}{2}\left(  \frac{x-\mu}{\sigma}\right)
^{2}\right)  +\frac{\beta}{\left(  \beta+1\right)  ^{3/2}}\right]  ,
\label{EQ:phi_2}%
\end{equation}
and
\begin{align}
\boldsymbol{K}_{\beta}\left(  \mu,\sigma\right)   &  =\mathrm{diag}%
\{K_{\beta,11}(\mu,\sigma),K_{\beta,22}(\mu,\sigma)\},\label{EQ:K}\\
K_{\beta,11}(\mu,\sigma)  &  =\frac{1}{\sigma^{2(\beta+1)}\left(  2\pi\right)
^{\beta}\left(  2\beta+1\right)  ^{3/2}},\nonumber\\
K_{\beta,22}(\mu,\sigma)  &  =\frac{\tau(\beta)}{\sigma^{2(\beta+1)}\left(
2\pi\right)  ^{\beta}},\nonumber\\
\tau(\beta)  &  =\frac{2\left(  2\beta^{2}+1\right)  \sqrt{2\beta+1}}%
{(2\beta+1)^{3}}-\frac{\beta^{2}}{(\beta+1)^{3}}, \label{tao}%
\end{align}
Making use of the above, the Rao-type test statistics finally turns out to be
\begin{align*}
R_{\beta,n}\left(  \mu_{0},\sigma_{0}\right)   &  =R_{\beta,n}\left(  \mu
_{0}\right)  +R_{\beta,n}\left(  \sigma_{0}\right)  ,\\
R_{\beta,n}\left(  \mu_{0}\right)   &  =\frac{1}{nK_{\beta,11}(\mu_{0}%
,\sigma_{0})}\left(  \sum_{i=1}^{n}u_{\beta,1}\left(  X_{i},\mu_{0},\sigma
_{0}\right)  \right)  ^{2},\\
R_{\beta,n}\left(  \sigma_{0}\right)   &  =\frac{1}{nK_{\beta,22}(\mu
_{0},\sigma_{0})}\left(  \sum_{i=1}^{n}u_{\beta,2}\left(  X_{i},\mu_{0}%
,\sigma_{0}\right)  \right)  ^{2},
\end{align*}
the combination of two one-dimensional Rao-type test statistics. On one hand,
$R_{\beta,n}\left(  \mu_{0}\right)  $ has the same expression as (\ref{R1}),
and is useful to test the hypothesis
\[
H_{0}:\mu=\mu_{0}\text{ against }H_{1}:\mu\neq\mu_{0}%
\]
from a normal population with known variance $\sigma_{0}^{2}$ and unknown mean
$\mu$. On the other hand
\[
R_{\beta,n}\left(  \sigma_{0}\right)  =\frac{1}{n\tau(\beta)}\left(
\sum_{i=1}^{n}\left[  \left(  \frac{X_{i}-\mu_{0}}{\sigma_{0}}\right)
^{2}-1\right]  \exp\left(  -\frac{\beta}{2}\left(  \frac{X_{i}-\mu_{0}}%
{\sigma_{0}}\right)  ^{2}\right)  +\frac{\beta}{\left(  \beta+1\right)
^{3/2}}\right)  ^{2},
\]
with $\tau(\beta)$ given by (\ref{tao}), has the same expression as the Rao
type test statistics to test the hypothesis
\[
H_{0}:\sigma^{2}=\sigma_{0}^{2}\text{ against }H_{1}:\sigma^{2}\neq\sigma
_{0}^{2}%
\]
from a normal population with known mean $\mu_{0}$ and unknown variance
$\sigma^{2}$. For $\beta=0$, we get
\begin{align}
R_{\beta=0,n}\left(  \mu_{0}\right)   &  =\left(  \frac{\bar{X}_{n}-\mu_{0}%
}{\frac{\sigma_{0}}{\sqrt{n}}}\right)  ^{2},~R_{\beta=0,n}\left(  \sigma
_{0}\right)  =\frac{n}{2}\left(  \frac{S_{\mu_{0}}^{2}-\sigma_{0}^{2}}%
{\sigma_{0}^{2}}\right)  ^{2},\nonumber\\
S_{\mu_{0}}^{2}  &  =\frac{1}{n}\sum_{i=1}^{n}\left(  X_{i}-\mu_{0}\right)
^{2}, \label{S2mu}%
\end{align}
so that the classical Rao test for testing the hypothesis (\ref{2.5}) is given
by the statistic
\[
R_{\beta=0,n}\left(  \mu_{0},\sigma_{0}\right)  =\left(  \frac{\bar{X}_{n}%
-\mu_{0}}{\frac{\sigma_{0}}{\sqrt{n}}}\right)  ^{2}+\frac{n}{2}\left(
\frac{S_{\mu_{0}}^{2}-\sigma_{0}^{2}}{\sigma_{0}^{2}}\right)  ^{2},
\]
which is the classical Rao test for this hypothesis $R_{n}\left(  \mu
_{0},\sigma_{0}\right)  $.

In the bottom panel of Figure \ref{figEx1} the value of Rao type test
statistics is plotted for telephone-line faults data for $\mu_{0}=0$ and
$\sigma_{0}=175$. The threshold for the acceptance region of the null
hypothesis shows that for outlier deleted data, i.e. when the first
observation is removed, the null hypothesis cannot be rejected with
$\alpha=0.05$, however for most of values $\beta<0.2$, including the classical
Rao test statistic ($\beta=0$), the null hypothesis is rejected for the full
data set. This different conclusion in the decision of the test clarifies the
lack of robustness of the classical Rao test statistic. On the other hand,
most of values $\beta>0.2$ have a similar value of the Rao type test statistic
deleting or not the outlier, this fact shows the robust property of the
proposed Rao type test statistics.


\subsection{Example 3 (Two-dimensional composite null
hypothesis)\label{Example3}}


Let $X$ be a normal population with unknown variance $\sigma^{2}$ and mean
$\mu.$ We will develop the Rao-type test statistics for testing%
\[
H_{0}:\mu=\mu_{0}\text{ against }H_{1}:\mu\neq\mu_{0},
\]
where $\sigma^{2}$ is an unknown nuisance parameter. In this case the full
parameter space is given by $\Theta=\{\left(  \mu,\sigma\right)  \in%
\mathbb{R}
\times%
\mathbb{R}
^{+}\}$, while that under the null hypothesis is given by $\Theta
_{0}=\{\left(  \mu,\sigma\right)  \in%
\mathbb{R}
\times%
\mathbb{R}
^{+}:\mu=\mu_{0}\}$, which is fixed in one of the two dimensions. If we
consider the function $m\left(  \mu,\sigma\right)  =\mu-\mu_{0},$ the null
hypothesis $H_{0}$ can be written alternatively as
%
\[
H_{0}:m(\mu,\sigma)=0.
\]
We can observe that in our case $\boldsymbol{M}\left(  \mu,\sigma\right)
=\left(  1,0\right)  ^{T}$. With $f_{\mu_{0},\sigma^{2}}(x)$ being the normal
density with mean $\mu_{0}$ and variance $\sigma^{2}$, we consider the
statistic given in Equation (\ref{redR}), with $r=1$ and $\boldsymbol{K}%
_{\beta}\left(  \mu,\sigma\right)  $ given by (\ref{EQ:K}). On the other hand
\[
\boldsymbol{U}_{\beta,n}(\mu_{0},\widetilde{\sigma}_{\beta})=\frac{1}%
{n}\left(  \sum_{i=1}^{n}u_{\beta,1}\left(  X_{i},\mu_{0},\widetilde{\sigma
}_{\beta}\right)  ,\sum_{i=1}^{n}u_{\beta,2}\left(  X_{i},\mu_{0}%
,\widetilde{\sigma}_{\beta}\right)  \right)  ^{T},
\]
with $u_{\beta,1}\left(  x,\mu,\sigma\right)  $ and $u_{\beta,2}\left(
x,\mu,\sigma\right)  $ being as given in Equations (\ref{EQ:phi_1}) and
(\ref{EQ:phi_2}), respectively. The estimator $\boldsymbol{\tilde{\theta}%
}_{\beta}=\left(  \mu_{0},\widetilde{\sigma}_{\beta}\right)  ^{T}$, for known
$\mu=\mu_{0}$ when $\beta>0$, is the solution of the nonlinear equation
$U_{2,\beta,n}(\mu_{0},\widetilde{\sigma}_{\beta})=0$, with $u_{2,\beta
}\left(  x,\mu,\sigma\right)  $ given $\mu=\mu_{0}$ has the same expression as
(\ref{EQ:phi_2}) and $U_{2,\beta,n}(\mu_{0},\widetilde{\sigma}_{\beta}%
)=\frac{1}{n}\sum_{i=1}^{n}u_{2,\beta}\left(  X_{i},\mu_{0},\sigma\right)  $.
Based on the previous calculations, we have
\begin{align*}
\tilde{R}_{\beta,n}(\mu_{0},\widetilde{\sigma}_{\beta})  &  =\frac{1}%
{n}K_{\beta,11}(\mu_{0},\widetilde{\sigma}_{\beta})\left(  \sum_{i=1}%
^{n}u_{\beta,1}\left(  X_{i},\mu_{0},\widetilde{\sigma}_{\beta}\right)
\right)  ^{2}\\
&  =\frac{1}{n}\left(  2\beta+1\right)  ^{\frac{3}{2}}\left(  \sum_{i=1}%
^{n}\frac{X_{i}-\mu_{0}}{\widetilde{\sigma}_{\beta}}\exp\left(  -\frac{\beta
}{2}\left(  \frac{X_{i}-\mu_{0}}{\widetilde{\sigma}_{\beta}}\right)
^{2}\right)  \right)  ^{2}.
\end{align*}
In particular, since $\widetilde{\sigma}_{\beta=0}^{2}=S_{\mu_{0}}^{2}$ given
by (\ref{S2mu}) for $\beta=0$, the classical Rao test statistic has the
following simple expression%
\[
\tilde{R}_{n}(\mu_{0},S_{\mu_{0}}^{2})=\tilde{R}_{\beta=0,n}(\mu
_{0},\widetilde{\sigma}_{\beta=0})=\left(  \frac{\bar{X}_{n}-\mu_{0}}%
{\frac{S_{\mu_{0}}}{\sqrt{n}}}\right)  ^{2}.
\]

In the top panel of Figure \ref{figEx2} the values of the Rao type test
statistics are plotted for the telephone-line faults data for $\mu_{0}=0$. For
the cleaned data, i.e., with the first observation removed, all Rao-type test
statistics are above the threshold and are able to reject the null hypothesis
at the nominal level $\alpha=0.05$. However, for the full data, most values of
$\beta<0.2$ produce test statistics smaller than the threshold and therefore
fail to reject the null hypothesis. The closeness of the statistics for the
full data and the outlier deleted data as a function of $\beta$ correspond
approximately to the closeness of the scale estimates for the full data and
the outlier deleted data which is demonstrated in the bottom panel of Figure 2.
On the whole this gives a good illustration of the robustness and stability
properties of the Rao-type test statistics for larger values of $\beta$.%

\begin{figure}[htbp]  \tabcolsep2.8pt  \centering
\begin{tabular}
[c]{l}%
{\includegraphics[
trim=1.382526in 0.775197in 1.038240in 0.349014in,
height=4.0032in,
width=6.4238in
]%
{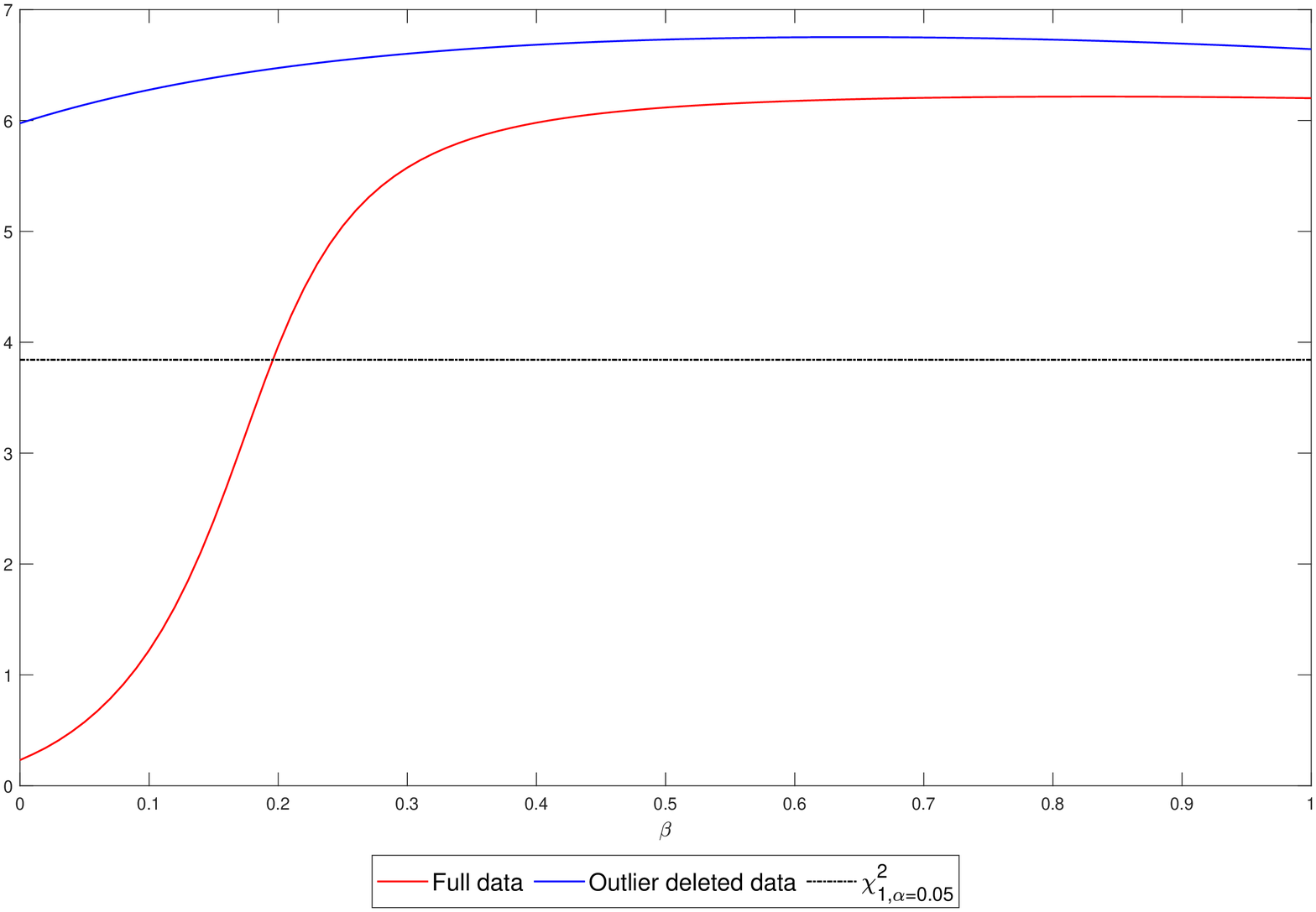}%
}
\\%
{\includegraphics[
trim=1.382526in 0.778705in 1.040930in 0.347260in,
height=4.0023in,
width=6.423in
]%
{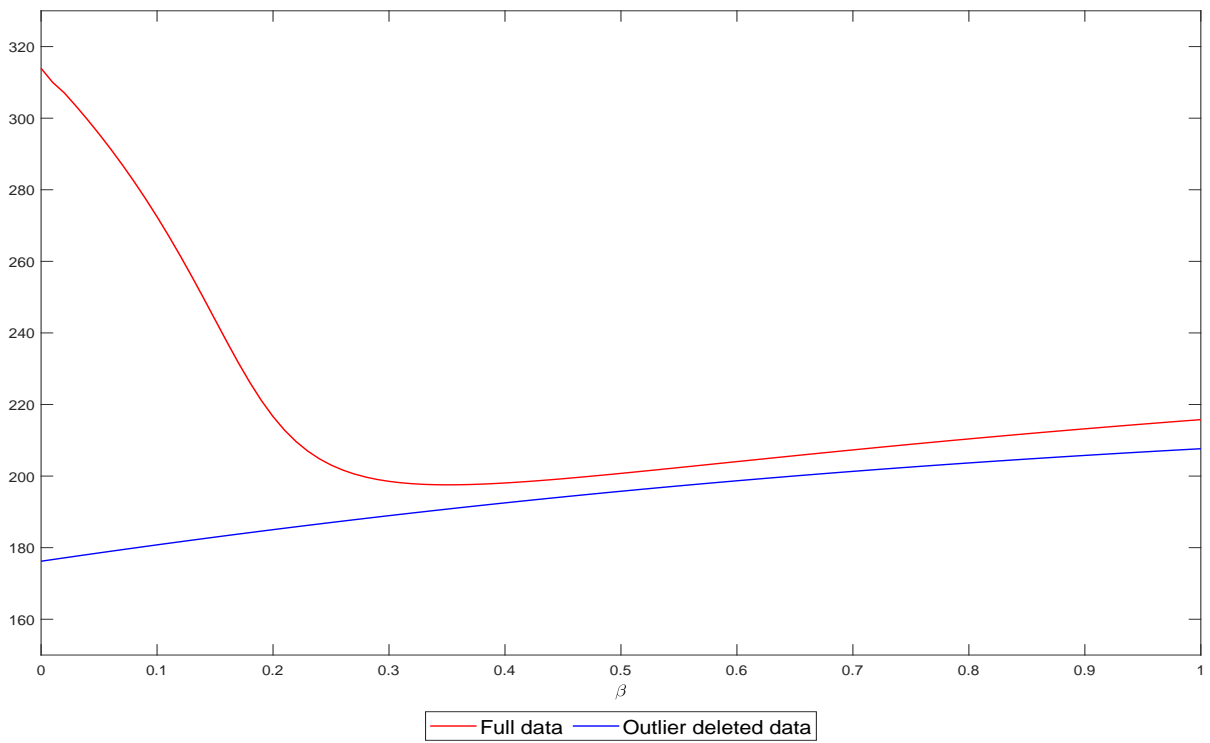}%
}
\end{tabular}
\caption{Telephone-line faults data: $\tilde{R}_{n,\beta}(\mu_{0},\widetilde{\sigma}_{\beta})$ (top) and $\widetilde{\sigma}_{\beta}$ (bottom) both in terms of tuning parameter $\beta$ on abscissa axis
		\label{figEx2}}
\end{figure}%

\section{Simulation study}

We have considered the one-dimensional simple null hypothesis and the
two-dimesnional composite null hypothesis for studying the performance of the
Rao-type test statistics for normal populations. Our simulations have covered
eleven values of the tuning parameter $\beta=k\times0.1$ for $k\in
\{0,1,...,10\}$; however in our plots we have only included seven values
associated to $k\in\{0,2,...,10\}\cup\{3\}$ to make the identification of
different curves, associated with $\beta$, easier. We have used $R=1,000,000$
replications and have taken sample sizes equal to $n\in\{5,6,...,50\}$; when
$10\%$ contamination is introduced, some of the smaller samples may not have
any outliers at all, but by the time samples are of sizes around $50$, and
reasonable stability in the general behavior of the Rao-type test statistics
under contamination should be observable under the two scenarios we are going
to describe.\medskip

\textbf{Scenario 1}: This is based on Example \ref{Example1} (one-dimensional
simple null hypothesis). Here $X$ is assumed to be a normal random variable
with known variance $\sigma_{0}^{2}=1$ and unknown mean $\mu$. The
corresponding Rao-type test statistics for testing%
\begin{equation}
H_{0}:\mu=0\text{ against }H_{1}:\mu\neq0,\label{eq0}%
\end{equation}
is given by%
\[
R_{\beta,n}\left(  \mu_{0}=0\right)  =\frac{1}{n}\left(  2\beta+1\right)
^{\frac{3}{2}}\left(  \sum_{i=1}^{n}X_{i}\exp\left(  -\frac{\beta}{2}X_{i}%
^{2}\right)  \right)  ^{2}.
\]
The empirical significance level of this test (at nominal level $\alpha$) is
computed as the proportion of replications (out of the total $R$) where the
Rao-type test statistic exceeds the asymptotic threshold given by
$\chi_{1,\alpha}^{2}$ ($\chi_{1}^{2}$ quantile of order $\alpha$, on the
right), with $\alpha=0.05$. For pure data it is seen, at the top panel of
Figure \ref{fig1s}, that the quantile of the classical Rao test statistic
($\beta=0$) matches almost perfectly with the chi-square quantile of order
$0.05$ for any sample size $n$, and the variation around the line
$\alpha=0.05$\ is due to the sampling fluctuations (smaller than $1\%$ in
absolute value) alone. As $n$ increases, the empirical levels of the Rao-type
test statistics (for all $\beta$) get closer to the nominal values. The
approximation is better as $\beta$ is closer to $0$, and hence $\beta=0.2$ has
better performance than $\beta=0.4$, but worse than $\beta=0$, in overall
terms for all sample sizes.

In order to study attained levels of our proposed tests the performance of the
significance levels under contamination, we generated observations from the
$0.9\mathcal{N}(0,1)+0.1\mathcal{N}(-4.5,1)$ mixture. It is observed that all
test statistics corresponding to $\beta\geq0.4$ provide remarkably stable
results in terms of the closeness to the empirical levels and nominal levels.
For very small values of $\beta$ the results are generally reasonable at small
sample sizes; however for $n>10$, the empirical level of the classical Rao
test blows up with increasing sample size. The same phenomenon is observed, at
a lesser degree, for the test corresponding to $\beta=0.2$. A slow inflation
appears to take place for $\beta=0.3$ as well.

To investigate the power behavior of the Rao-type test statistics under pure
data, we have taken $\mathcal{N}(-0.5,1)$ for all simulations; the empirical
power of a given test is the empirical proportion of the number of test
statistics exceeding the chi-square quantile threshold, $\chi_{1,\alpha}^{2}$.
As shown in the top panel of Figure \ref{fig2s}, the classical Rao test
statistic ($\beta=0$) exhibits the highest power with pure data and in general
as $\beta$ decreases the power is higher. For creating contamination, the
samples for the scenario described in the bottom panel of Figure 4 come from
the normal mixture $0.9\mathcal{N}(-0.5,1)+0.1\mathcal{N}(5,1)$. Under
contaminated data the classical Rao test statistic ($\beta=0$) exhibits a very
significant drop in power, making the power curve practically flat. Most of
the test statistics corresponding to positive values of $\beta$ perform much
better in holding the power levels under contamination. The best performance
in terms of power under contamination is provided by the test-statistic
corresponding to $\beta=0.4$.\medskip

\textbf{Scenario 2}: Taking Example \ref{Example3} as the basis
(two-dimensional composite null hypothesis), we have taken $X$ to be a
normally distributed random variable with unknown mean $\mu$ and variance
$\sigma^{2}$. The corresponding Rao-type test statistics for testing the
hypothesis in (\ref{eq0}) is given by%
\[
\tilde{R}_{\beta,n}\left(  \mu_{0}=0,\widetilde{\sigma}_{\beta}^{2}\right)
=\frac{1}{n}\left(  2\beta+1\right)  ^{\frac{3}{2}}\frac{1}{\widetilde{\sigma
}_{\beta}^{2}}\left(  \sum_{i=1}^{n}X_{i}\exp\left(  -\frac{\beta}{2}%
\frac{X_{i}^{2}}{\widetilde{\sigma}_{\beta}^{2}}\right)  \right)  ^{2},
\]
where $\widetilde{\sigma}_{\beta}^{2}$ is the estimation of $\sigma^{2}$ under
the assumption that $\mu_{0}=0$. The scheme of contamination and outliers is
exactly the same as for Scenario 1, but the results shown in Scenario 2, by
Figures \ref{fig1c}-\ref{fig2c}, are very different. Surprisingly, for pure
data the ordinary Rao test is conservative, particularly in small samples, and the nominal levels are better approximated with
$\beta>0.2$ than for $\beta=0$, and under contamination the ordinary Rao test simply breaks down, while all the others reasonably hold their level; even the $\beta = 0.2$ test has an observed level less than 0.1 at a sample size of $n = 50$  (Figure \ref{fig1c}). Power calculations for pure data indicate that the ordinary Rao test has high power except at very low sample sizes (a consequence of its conservative nature), but its power practically vanishes under contamination, where values of $\beta$ between 0.6 and 0.8 appear to have the best performance  (Figure \ref{fig2c}).%

\begin{figure}[htbp]  \tabcolsep2.8pt  \centering
\begin{tabular}
[c]{l}%
\raisebox{-0cm}{\includegraphics[
trim=1.255760in 1.436243in 1.079013in 0.432292in,
height=9.1204cm,
width=16.7207cm
]%
{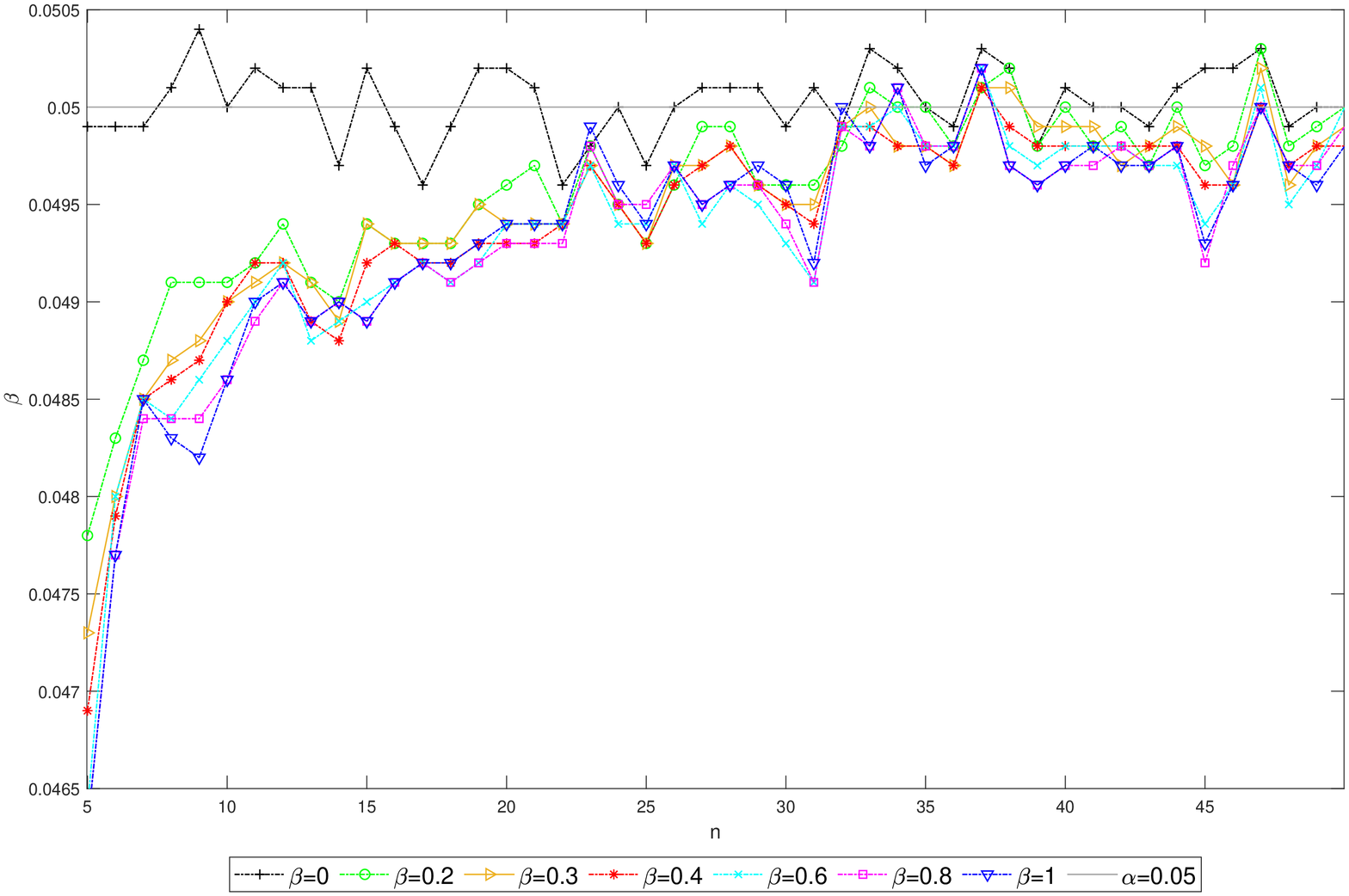}%
}
\\%
\raisebox{-0cm}{\includegraphics[
trim=1.257185in 0.538278in 1.077588in 0.431457in,
height=10.3768cm,
width=16.7207cm
]%
{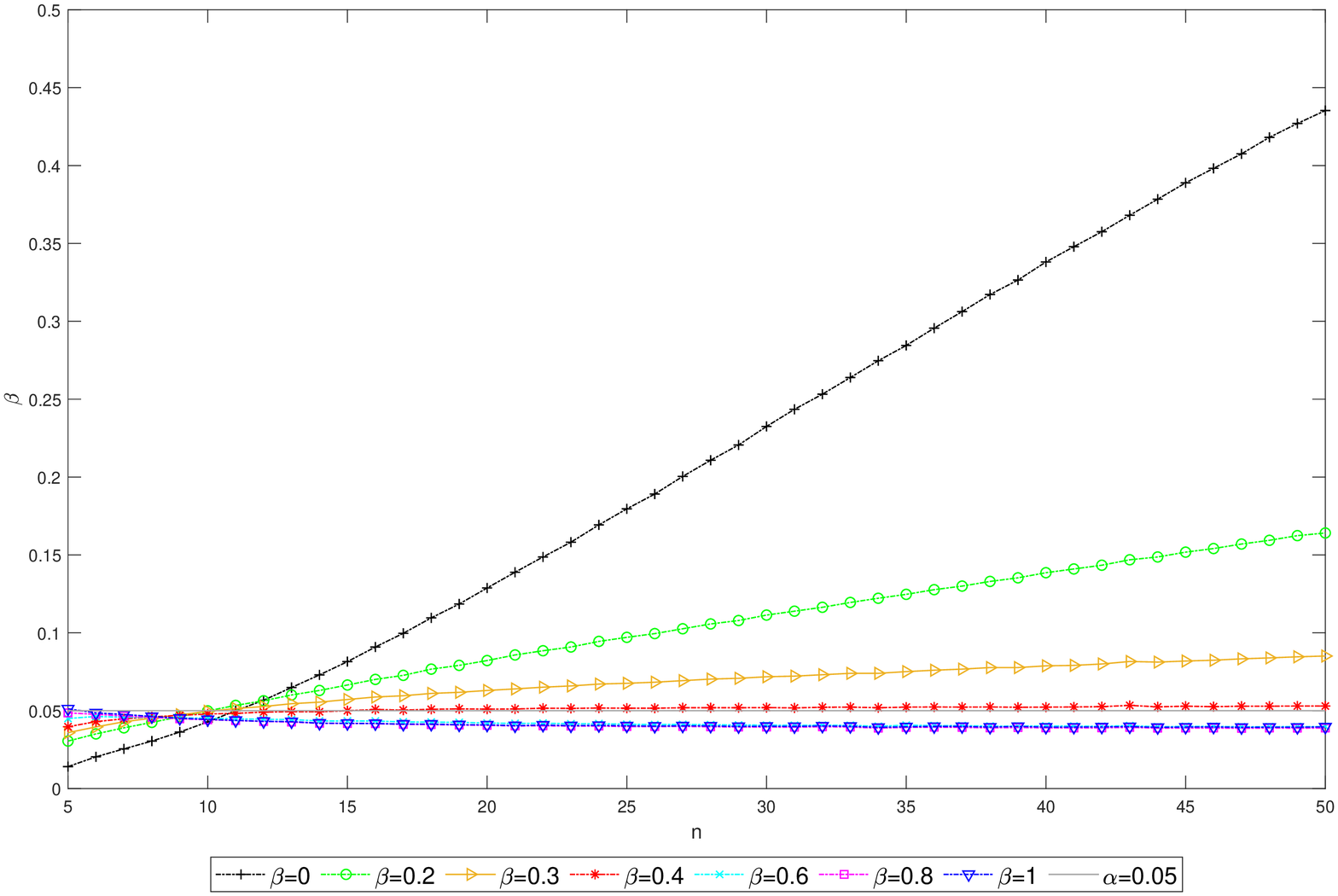}%
}
\end{tabular}
\caption{Simulated significance levels in simple null hypothesis for pure data ($\epsilon=0$, top) and contaminated data ($\epsilon=0.1$, bottom)
		\label{fig1s}}
\end{figure}%
%

\begin{figure}[htbp]  \tabcolsep2.8pt  \centering
\begin{tabular}
[c]{l}%
\raisebox{-0cm}{\includegraphics[
trim=1.255760in 1.436243in 1.079013in 0.432292in,
height=9.1204cm,
width=16.7207cm
]%
{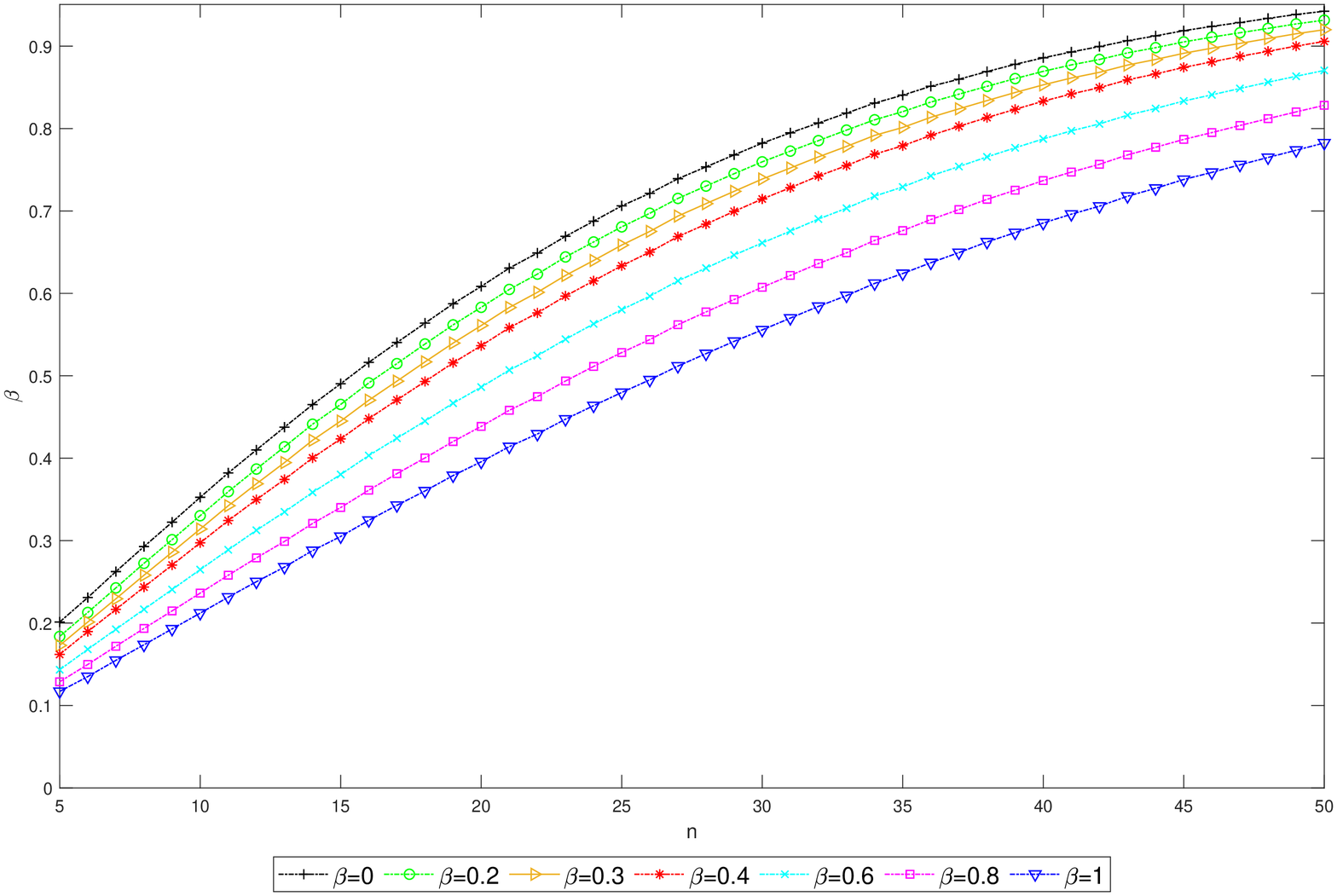}%
}
\\%
\raisebox{-0cm}{\includegraphics[
trim=1.257185in 0.538278in 1.077588in 0.431457in,
height=10.3768cm,
width=16.7185cm
]%
{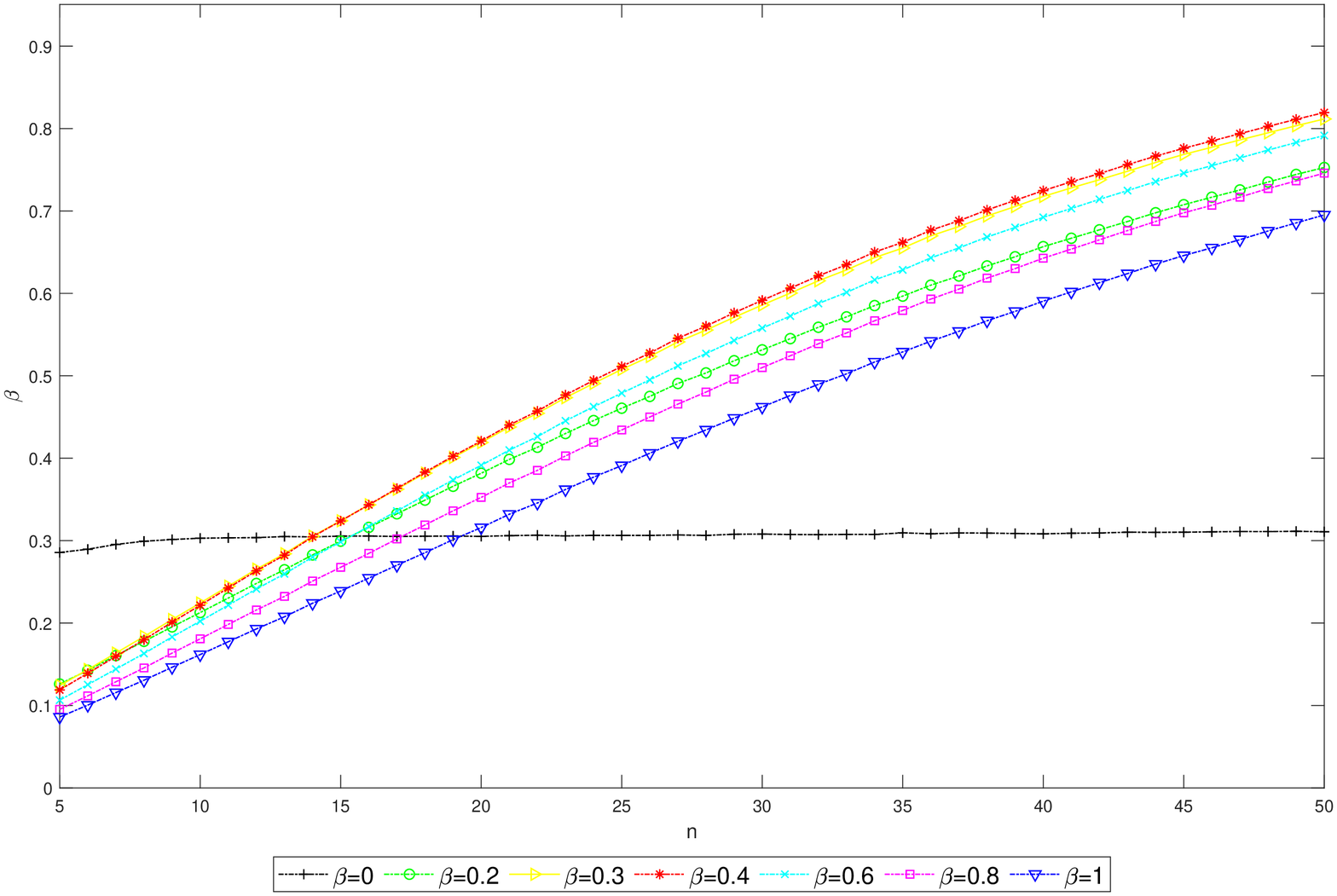}%
}
\end{tabular}
\caption{Simulated powers in simple null hypothesis for pure data ($\epsilon=1$, top) and contaminated data ($\epsilon=0.9$, bottom)
		\label{fig2s}}
\end{figure}%
%

\begin{figure}[htbp]  \tabcolsep2.8pt  \centering
\begin{tabular}
[c]{l}%
\raisebox{-0cm}{\includegraphics[
trim=1.255760in 1.436243in 1.079013in 0.432292in,
height=9.1204cm,
width=16.7207cm
]%
{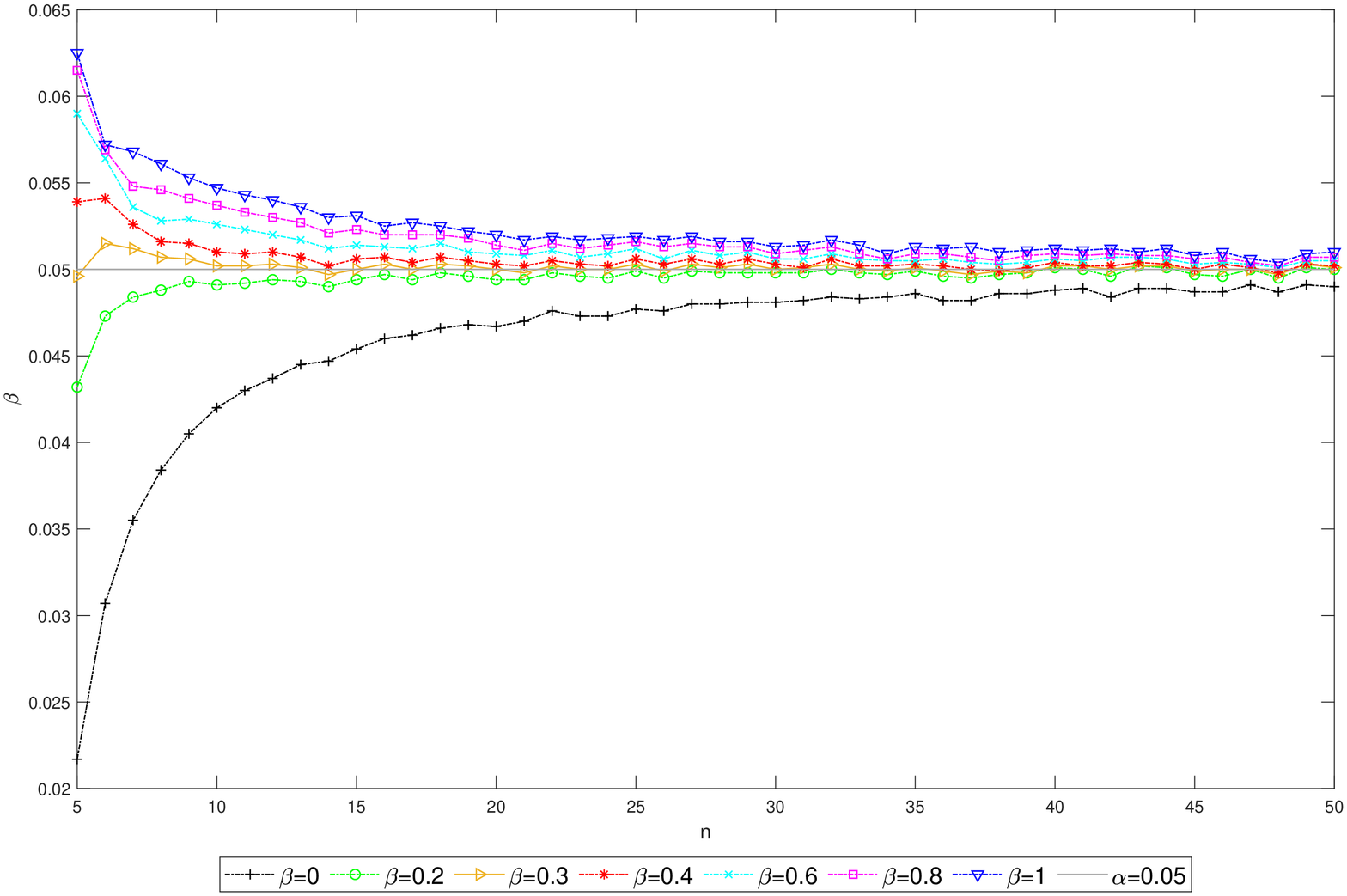}%
}
\\%
\raisebox{-0cm}{\includegraphics[
trim=1.257185in 0.538278in 1.077588in 0.431457in,
height=10.3768cm,
width=16.7207cm
]%
{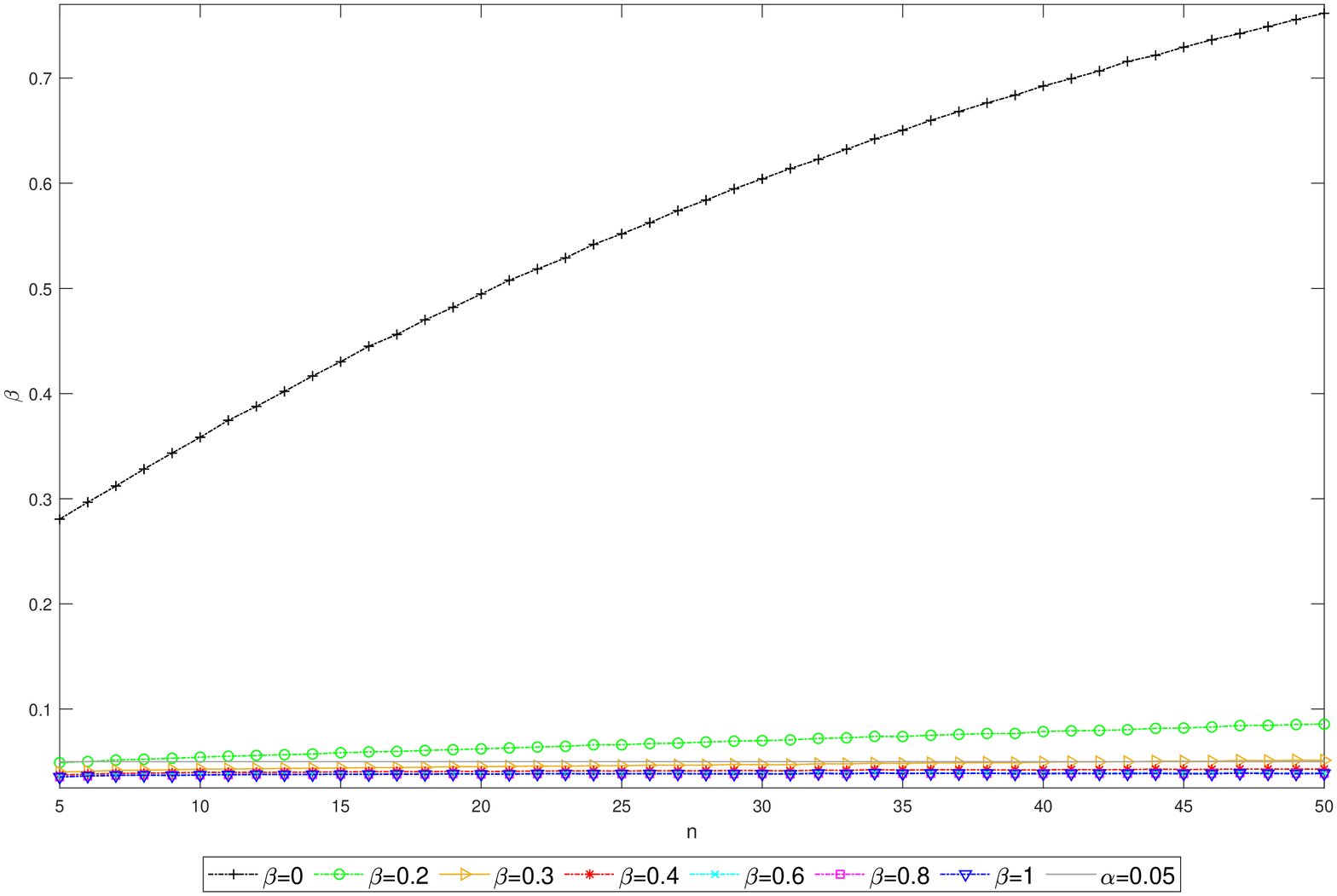}%
}
\end{tabular}
\caption{Simulated significance levels in composite null hypothesis for pure data ($\epsilon=0$, top) and contaminated data ($\epsilon=0.1$, bottom)
		\label{fig1c}}
\end{figure}%
%

\begin{figure}[htbp]  \tabcolsep2.8pt  \centering
\begin{tabular}
[c]{l}%
\raisebox{-0cm}{\includegraphics[
trim=1.255760in 1.436243in 1.079013in 0.432292in,
height=9.1204cm,
width=16.7207cm
]%
{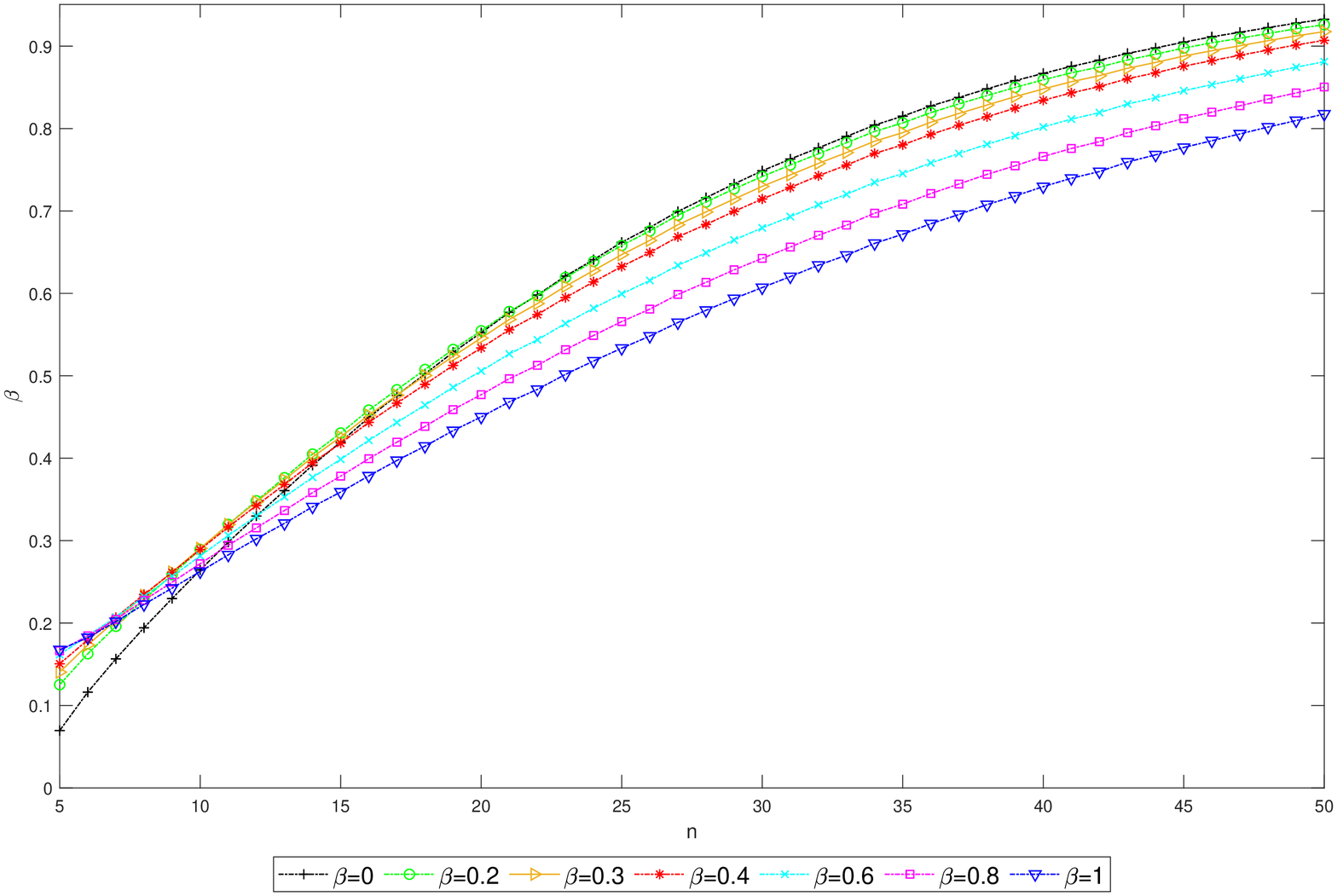}%
}
\\%
\raisebox{-0cm}{\includegraphics[
trim=1.257185in 0.538278in 1.077588in 0.431457in,
height=10.3768cm,
width=16.7207cm
]%
{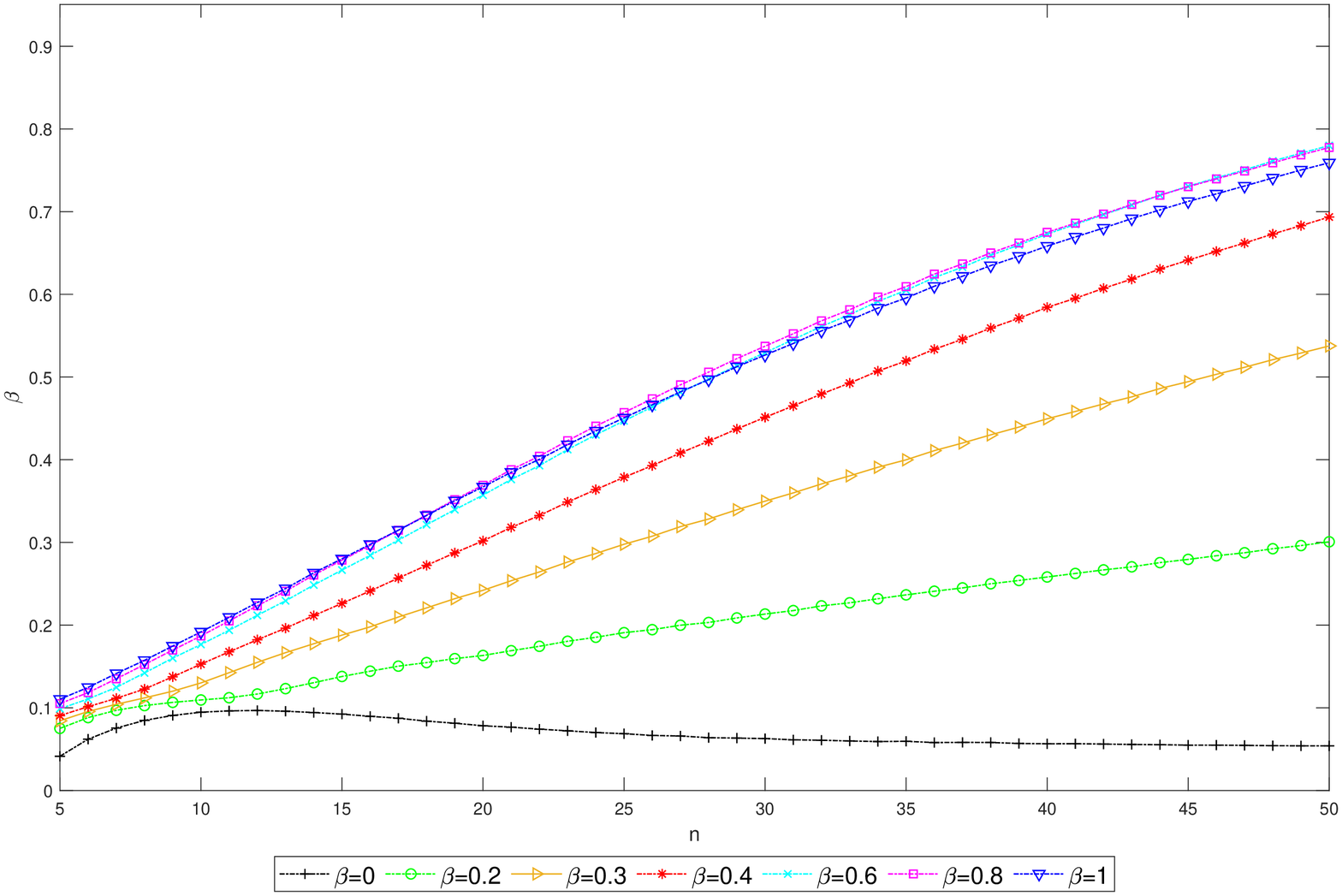}%
}
\end{tabular}
\caption{Simulated powers in composite null hypothesis for pure data ($\epsilon=1$, top) and contaminated data ($\epsilon=0.9$, bottom)
		\label{fig2c}}
\end{figure}%

\newpage
\section{Concluding Remarks}

The traditional Rao test is a very useful tool in applied statistics and econometrics. However, as the default application of tests of this type are based the maximum likelihood estimator and the likelihood score function, lack of robustness and sensitivity to outliers is an inherent problem in the traditional Rao test. In this paper we have proposed a c;lass of Rao type tests, including the ordinary Rao test, many members of which have remarkable robustness properties often with very little loss in power relative to the ordinary Rao test. In i.i.d. data the application of the test is simple, and we believe that the theoretical results and the numerical support demonstrate that the class of tests help to cover the primary deficiency of the traditional Rao test at a very little cost.

\appendix



\begin{thebibliography}{99}
\bibitem {pear}Pearson, K. (1900). On the criterion that a given system of
deviations from the probable in the case of a correlated system of variables
is such that it can be reasonably supposed to have arisen from random sampling.
\textit{Philosophical Magazine, }\textbf{5}, 157--175.
\end{thebibliography}
\end{document}